\documentclass[11pt, a4paper]{article}
\usepackage[margin=1in]{geometry}

\usepackage{microtype}
\usepackage{graphicx}
\usepackage{subcaption}
\usepackage{booktabs} 
\usepackage{amsmath}
\usepackage{amsthm}
\usepackage{amssymb}
\usepackage{wrapfig}
\usepackage{graphicx}
\usepackage[dvipsnames]{xcolor}  
\usepackage{multirow}
\usepackage{epsfig,amssymb,amsfonts,amsmath,amsthm} 
\usepackage{tikzit}

\tikzstyle{basic}=[fill=white, draw=black, shape=circle]
\tikzstyle{square}=[fill=white, draw=black, shape=rectangle]
\tikzstyle{big dashed}=[fill=white, draw=black, shape=circle, minimum width=1cm, dashed]
\tikzstyle{vertical ellipse dashed}=[fill=none, draw=blue, minimum width=0.75cm, minimum height=3cm, ellipse, dashed, tikzit shape=rectangle, tikzit draw=blue, tikzit fill=white]
\tikzstyle{small vertical ellipse dashed}=[fill=none, draw=blue, shape=circle, tikzit fill=white, tikzit draw=blue, dashed, minimum width=0.75cm, minimum height=1.5cm, tikzit shape=rectangle, ellipse]
\tikzstyle{tiny vertical ellipse dashed}=[fill=none, draw=blue, shape=circle, tikzit fill=white, ellipse, dashed, minimum width=0.75cm, minimum height=1cm, tikzit shape=rectangle]
\tikzstyle{red}=[fill=red, draw=black, shape=circle]
\tikzstyle{green}=[fill={rgb,255: red,0; green,128; blue,128}, draw=black, shape=circle]
\tikzstyle{blue}=[fill=blue, draw=black, shape=circle]
\tikzstyle{huge dashed}=[fill=white, draw=black, shape=circle, dashed, minimum width=2cm]
\tikzstyle{medium}=[fill=white, draw=black, shape=circle, minimum width=1cm]
\tikzstyle{pale green}=[fill={rgb,255: red,173; green,231; blue,0}, draw=black, shape=circle, minimum width=1cm]
\tikzstyle{horizontal ellipse dashed}=[fill=white, draw=black, tikzit draw=magenta, tikzit shape=rectangle, minimum width=3cm, minimum height=0.75cm, ellipse, dashed]
\tikzstyle{minsize}=[fill=white, draw=black, shape=circle, minimum width=0.75cm]
\tikzstyle{horizontal ellipse green}=[fill={rgb,255: red,191; green,255; blue,0}, draw=black, tikzit draw={rgb,255: red,191; green,255; blue,0}, tikzit shape=rectangle, minimum width=3cm, minimum height=0.75cm, ellipse, dashed]
\tikzstyle{horizontal ellipse blue}=[fill={rgb,255: red,107; green,203; blue,255}, draw=black, tikzit draw=blue, tikzit shape=rectangle, minimum width=3cm, minimum height=0.75cm, ellipse, dashed]
\tikzstyle{smallblack}=[fill=black, draw=black, shape=circle, inner sep=0 pt, minimum size=3 pt]
\tikzstyle{smallSquare}=[fill=white, draw=black, shape=rectangle, inner sep=0 pt, minimum size=6 pt]
\tikzstyle{smallCircle}=[fill=white, draw=black, shape=circle, inner sep=0 pt, minimum size=6 pt]
\tikzstyle{big vertical ellipse dashed}=[fill=none, draw=blue, shape=circle, tikzit shape=rectangle, ellipse, dashed, minimum width=0.95cm, minimum height=3.7cm]
\tikzstyle{smallred}=[fill=red, draw=red, shape=circle, inner sep=0 pt, minimum size=3 pt]
\tikzstyle{smallblue}=[fill=blue, draw=blue, shape=circle, inner sep=0pt, minimum size=3pt]

\tikzstyle{directed}=[->]
\tikzstyle{undirected}=[-, line width=1pt]
\tikzstyle{directed red}=[draw=red, ->, line width=1pt]
\tikzstyle{directed green}=[draw={rgb,255: red,0; green,128; blue,128}, ->, line width=1pt]
\tikzstyle{directed blue}=[draw=blue, ->, line width=1pt]
\tikzstyle{directed purple}=[draw={rgb,255: red,128; green,0; blue,128}, ->, line width=1pt]
\tikzstyle{undirected red}=[-, draw=red, line width=1pt]
\tikzstyle{undirected green}=[-, draw={rgb,255: red,0; green,107; blue,61}, line width=1pt]
\tikzstyle{undirected blue}=[-, draw=blue, line width=1pt]
\tikzstyle{undirected purple}=[-, draw={rgb,255: red,128; green,0; blue,128}, line width=1pt]
\tikzstyle{undirected dashed}=[-, line width=1pt, dashed]
\tikzstyle{orange dashed}=[-, draw={rgb,255: red,255; green,128; blue,0}, dashed, line width=1.5pt]
\tikzstyle{directed dash}=[->, dashed]
\tikzstyle{blue dashed}=[-, draw=blue, dashed, line width=1pt]
\tikzstyle{green dashed}=[-, draw={rgb,255: red,0; green,162; blue,0}, dashed, line width=1pt]
\tikzstyle{blue filled}=[-, fill={blue!20}, draw=blue, line width=1pt, opacity=0.5, tikzit fill=white]
\tikzstyle{red filled}=[-, fill={red!20}, line width=1pt, draw=red, opacity=0.5, tikzit fill=white]
\tikzstyle{green filled}=[-, line width=1pt, draw={rgb,255: red,0; green,107; blue,61}, opacity=0.5, tikzit fill=white, fill={rgb,255: red,149; green,255; blue,179}]
\tikzstyle{orange filled}=[-, fill={orange!20}, draw=orange, line width=1pt, opacity=0.5, tikzit fill=white]
\tikzstyle{undirected dashed}=[-, draw=black, dashed, line width=1pt]

\usepackage{appendix}
\usepackage[ruled,vlined,linesnumbered]{algorithm2e}
\usepackage{hhline} 
\usepackage{hyperref}
\usepackage{times}
\usepackage{placeins}
\usepackage{bm}

\newtheorem{theorem}{Theorem} 

\newtheorem{lemma}{Lemma}[section]

\newtheorem{definition}{Definition}

\newtheorem{example}{Example}

\newcommand{\norm}[1]{\left\| #1\right\|}                  
\newcommand{\abs}[1]{\left\lvert#1\right\rvert}
\DeclareMathOperator*{\argmin}{arg\,min}        

\newcommand{\cond}{\Phi}
\newcommand{\vol}{\mathrm{vol}}

\newcommand{\bigo}[1]{O\!\left(#1\right)}
\newcommand{\bigomega}[1]{\Omega\!\left(#1\right)}
\newcommand{\polylog}{\mathrm{polylog}}

\newcommand{\R}{\mathbb{R}}

\newcommand{\Z}{\mathbb{Z}}
\newcommand{\union}{\cup}
\newcommand{\intersect}{\cap}
\newcommand{\cardinality}[1]{\abs{#1}}

\newcommand{\barx}{\Bar{\vecx}}
\newcommand{\barxi}{\barx^{(i)}}
\newcommand{\barxj}{\barx^{(j)}}
\newcommand{\barg}{\Bar{\vecg}}
\newcommand{\hatg}{\widehat{\vecg}}
\newcommand{\hatf}{\widehat{\vecf}}

\newcommand{\geqvewg}{\graphg = (\vertexsetg, \edgesetg, \weight_\graphg)} 

\definecolor{indiagreen}{rgb}{0.07, 0.53, 0.03}

\newcommand{\twopartdef}[4]
{
	\left\{
		\begin{array}{ll}
			#1 & \mbox{if } #2 \\
			#3 & \mbox{if } #4
		\end{array}
	\right.
}

\newcommand{\allnotation}[1]{#1}

\renewcommand{\vec}[1]{{\allnotation{#1}}}
\newcommand{\vecf}{\vec{f}}
\newcommand{\vecg}{\vec{g}}
\newcommand{\vecx}{\vec{x}}

\newcommand{\vecu}{\vec{u}}
\newcommand{\vecv}{\vec{v}}

\newcommand{\vecp}{\vec{p}}

\newcommand{\veca}{\vec{a}}
\newcommand{\vecb}{\vec{b}}

\newcommand{\vecc}{\vec{c}}

\newcommand{\indicatorvec}{\vec{\chi}}

\newcommand{\graph}[1]{{\allnotation{#1}}}
\newcommand{\graphg}{\graph{G}}

\newcommand{\set}[1]{{\allnotation{#1}}}
\newcommand{\setv}{\set{V}}
\newcommand{\sete}{\set{E}}
\newcommand{\sets}{\set{S}}
\newcommand{\sett}{\set{T}}

\newcommand{\seta}{\set{A}}
\newcommand{\setb}{\set{B}}

\newcommand{\setm}{\set{M}}
\newcommand{\vertexset}{\setv}
\newcommand{\vertexsetg}{\vertexset_\graphg}

\newcommand{\edgeset}{\sete}
\newcommand{\edgesetg}{\edgeset_\graphg}

\newcommand{\mat}[1]{{\allnotation{#1}}}

\newcommand{\transpose}{\intercal}

\newcommand{\lapn}{\mat{\mathcal{N}}}
\newcommand{\degm}{\mat{D}}
\newcommand{\degmhalf}{\degm^{\allnotation{\frac{1}{2}}}}
\newcommand{\degmhalfneg}{\degm^{\allnotation{-\frac{1}{2}}}}
\newcommand{\adj}{\mat{A}}
\newcommand{\adjn}{\mat{\mathcal{A}}}

\newcommand{\identity}{\mat{I}}

\renewcommand{\deg}{{\allnotation{d}}}
\newcommand{\weight}{{\allnotation{w}}}

\title{A Tighter Analysis of Spectral Clustering, and Beyond\footnote{A preliminary version of this work appeared at   ICML~2022. This work is supported by a Langmuir PhD Scholarship, and an EPSRC Early Career Fellowship~(EP/T00729X/1).}}

\author{%
Peter Macgregor\\ University of Edinburgh
\and
He Sun\\ 
University of Edinburgh}
\date{}

\numberwithin{equation}{section}

\begin{document}

\maketitle

\begin{abstract}
This work studies the classical spectral clustering algorithm which embeds the vertices of some graph $G=(V_G, E_G)$ into $\mathbb{R}^k$  using $k$ eigenvectors of some matrix of $G$, and applies $k$-means to partition $V_G$ into $k$ clusters. Our first result is a tighter analysis   on the performance of spectral clustering,  and explains  why it works   under some much weaker condition than the ones studied in the literature.  For the second result, we show that,
by applying fewer than $k$ eigenvectors to construct the embedding, spectral clustering
is able to produce better output
for many practical instances; this result is the first of its kind in spectral clustering.    Besides its conceptual and theoretical significance, the practical impact of our work is  demonstrated by the empirical analysis on both synthetic and real-world datasets, in which spectral clustering produces comparable or better results with fewer than $k$ eigenvectors.
\end{abstract}

\newcommand{\graphm}{\graph{M}}
\newcommand{\peye}{\vecp^{(i)}}
\newcommand{\pj}{\vecp^{(j)}}
\newcommand{\eps}{\epsilon}
\newcommand{\APT}{\textsf{APT}}

\begin{figure*}[t]
    \centering
    \begin{subfigure}{0.3\textwidth} 
    \includegraphics[width=\textwidth]{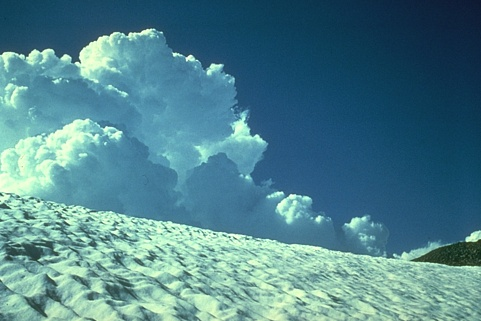}
    \caption{Original Image}
    \end{subfigure}
    \hspace{1em}
    \begin{subfigure}{0.3\textwidth} 
    \includegraphics[width=\textwidth]{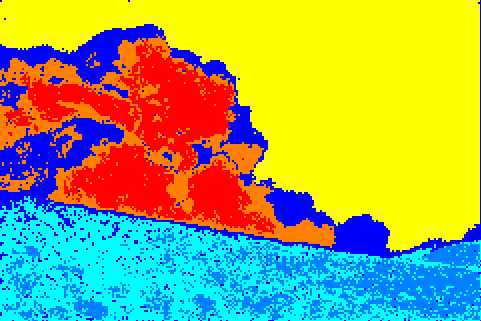}
    \caption{$6$ clusters with $3$ vectors}
    \end{subfigure}
    \hspace{1em}
    \begin{subfigure}{0.3\textwidth} 
    \includegraphics[width=\textwidth]{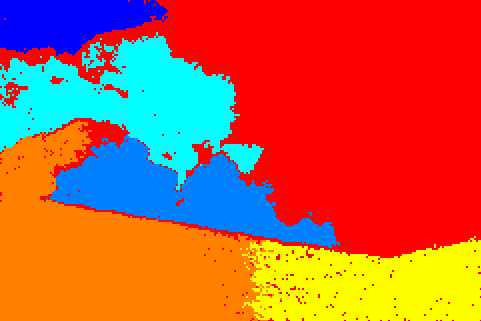}
    \caption{$6$ clusters with $6$ vectors}
    \end{subfigure}
    \par\bigskip
    \begin{subfigure}{0.3\textwidth} 
    \includegraphics[width=\textwidth]{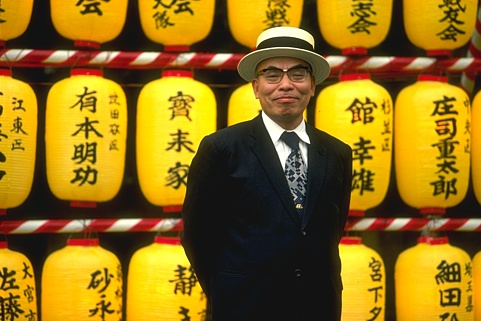}
    \caption{Original Image}
    \end{subfigure}
    \hspace{1em}
    \begin{subfigure}{0.3\textwidth} 
    \includegraphics[width=\textwidth]{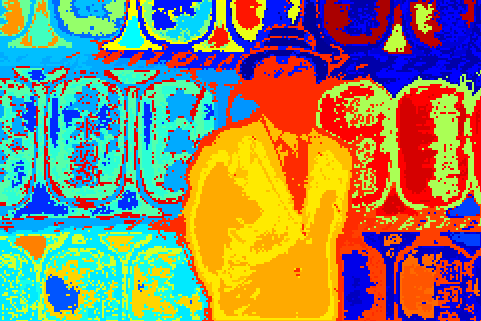}
    \caption{$45$ clusters with $7$ vectors}
    \end{subfigure}
    \hspace{1em}
    \begin{subfigure}{0.3\textwidth} 
    \includegraphics[width=\textwidth]{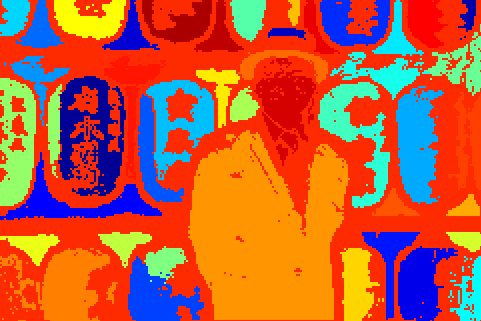}
    \caption{$45$ clusters with $45$ vectors}
    \end{subfigure}
    \caption[Examples of image segmentation using spectral clustering]{Examples of image segmentation using spectral clustering; the original images are from the BSDS.  The Rand Index of segmentation (b) is $0.83$, while (c) has Rand Index $0.78$. Segmentation (e) has Rand Index $0.92$, and (f) has Rand Index $0.80$. Hence, it's clear that spectral clustering with fewer than $k$ eigenvectors suffices to produce comparable or better output.} 
    \label{fig:bsds_results_intro}
\end{figure*}
 
\section{Introduction\label{sec:introduction}}

Graph clustering is a fundamental  problem in unsupervised learning, and has  comprehensive applications in computer science and related scientific fields. Among various techniques to solve   graph clustering problems, spectral clustering is probably the easiest one to implement, and has been widely applied in practice.  Spectral clustering  can be  easily described as follows: for any graph $\graphg =(\vertexset_\graphg, \edgeset_\graphg)$ and some $k \in \Z^+$ as input, spectral clustering   embeds the vertices of $\vertexset_\graphg$ into $\R^k$ based on the bottom $k$ eigenvectors of the Laplacian matrix of $\graphg$, and employs $k$-means on the embedded points to partition $\vertexset_\graphg$ into $k$ clusters.  Thanks to its simplicity and excellent performance in practice, spectral clustering has been widely applied over the past three decades~\cite{spielmanSpectralPartitioningWorks1996}.
 
In this work we study spectral clustering, and present two results.
Our first result is a tighter analysis of spectral clustering for well-clustered graphs.
Informally, we analyse the performance guarantee of spectral clustering under a simple assumption\footnote{This assumption will be formally defined in Section~\ref{sec:structure}.} on the input graph. While all the previous work~(e.g., \cite{leeMultiwaySpectralPartitioning2014,kolevNoteSpectralClustering2016,mizutaniImprovedAnalysisSpectral2021,pengPartitioningWellClusteredGraphs2017}) on the same problem suggests that the assumption on the input graph must depend on $k$, our result demonstrates that the performance of spectral clustering can be rigorously analysed under a general condition independent of $k$.
To the best of our knowledge, our work presents the first result of its kind, and hence we believe that this result and the novel analysis used in its proof are important, and might have further applications in graph clustering. 
  
Secondly, we study the clustering problem in which the crossing edges between the optimal clusters $\{\sets_i\}_{i=1}^k$ present some noticeable pattern, which we call the \emph{meta-graph} in this work. Notice that, when viewing every cluster $\sets_i$ as a ``giant vertex'', our meta-graph captures the intrinsic connection between the optimal clusters, and could be significantly different from a clique graph. We prove that, when this is the case, one can simply apply classical spectral clustering while employing fewer than $k$ eigenvectors to construct the embedding and, surprisingly, this will produce a better clustering result.
The significance of this result is further demonstrated by our extensive experimental analysis on the well-known BSDS, MNIST, and USPS datasets~\cite{arbelaezContourDetectionHierarchical2011, lecunGradientbasedLearningApplied1998, hullDatabaseHandwrittenText1994}. While we discuss the experimental details in Section~\ref{sec:metaExperiments}, the performance of our algorithm is showcased in Figure~\ref{fig:bsds_results_intro}: in order to find $6$ and $45$ clusters, spectral clustering with $3$ and $7$ eigenvectors produce better results than the ones with $6$ and $45$ eigenvectors according to the default metric of the BSDS dataset.

\paragraph{Related work.}
Our first result on the analysis of spectral clustering is tightly related to a number of research that analyses spectral clustering algorithms  under various conditions~(e.g., \cite{leeMultiwaySpectralPartitioning2014, kolevNoteSpectralClustering2016, mizutaniImprovedAnalysisSpectral2021,gharanPartitioningExpanders2014, ngSpectralClusteringAnalysis2001, pengPartitioningWellClusteredGraphs2017}).
While we compare in detail between these works and ours in later sections,   to the best of our knowledge,  our work presents the first result proving   spectral clustering works under some general condition independent of $n$ and $k$.  Our work is also related to studies on designing local, and distributed clustering algorithms based on different assumptions~(e.g.,~\cite{czumajTestingClusterStructure2015, orecchiaFlowbasedAlgorithmsLocal2014, zhuLocalAlgorithmFinding2013});
due to limited computational resources available, these works require stronger assumptions on input graphs than ours. 

Our second result on spectral clustering with fewer eigenvalues   is   linked  to 
efficient spectral algorithms to find cluster-structures.  While it's known that flow and path structures of clusters in  digraphs can be uncovered with complex-valued Hermitian  matrices~\cite{cucuringuHermitianMatricesClustering2020, laenenHigherOrderSpectralClustering2020}, our work shows that one can apply  real-valued   Laplacians of undirected graphs, and find more general  patterns of clusters characterised by our structure theorem.
Rebagliati and Verri~\cite{rebagliatiSpectralClusteringMore2011} propose using more than $k$ eigenvectors for spectral clustering, although their assumptions on the input graph are different to ours and so the result is not directly comparable.

\section{Preliminaries}
Let $\graphg=(\vertexset_\graphg,\edgeset_\graphg,\weight)$ be an undirected graph with $n$ vertices, $m$ edges, and weight function $\weight: \vertexset_\graphg \times \vertexset_\graphg \rightarrow \R_{\geq 0}$.  
For any edge $e = \{u, v\} \in \edgeset_\graphg$, we write the weight of $\{u,v\}$ by $\weight_{uv}$ or $\weight_e$. 
For a vertex $u \in \vertexset_\graphg$, we denote its \emph{degree} by $\deg_\graphg(u) \triangleq \sum_{v \in \vertexset} \weight_{uv}$. 
For any two sets $\sets, \sett \subset \vertexset_\graphg$, we define the \emph{cut value} $\weight(\sets, \sett) \triangleq \sum_{e \in \edgeset_\graphg(\sets, \sett)} \weight_e$, where $\edgeset_\graphg(\sets, \sett)$ is the set of edges between $\sets$ and $\sett$.
For any set $\sets \subseteq \vertexset_\graphg$, the \emph{volume} of $\sets$ is $\vol_\graphg(\sets)\triangleq \sum_{u\in \sets} \deg_\graphg(u)$,
and we write $\vol(\graphg)$ when referring to $\vol(\vertexset_\graphg)$. For any nonempty subset $\sets \subseteq \vertexset_\graphg$, 
we define the \emph{conductance} of $\sets$ by 
\[
\cond_\graphg(\sets) \triangleq \frac{\weight(\sets, \vertexset \setminus \sets )}{\vol_\graphg(\sets)}.
\]
Furthermore, we define the conductance of the graph $\graphg$ by
\[
 \cond_\graphg\triangleq \min_{\substack{\sets \subset \vertexset\\ \vol(\sets) \leq \vol(\vertexset)/2}} \cond_\graphg(\sets).
\]
We call subsets of vertices $\seta_1,\ldots, \seta_k$ a \emph{k-way partition} of $\graphg$ if $\seta_i\cap \seta_j=\emptyset$ for different $i$ and $j$, and $\bigcup_{i=1}^k \seta_i = \vertexset$. Generalising the definition of conductance,  we define \emph{$k$-way expansion constant} by
\[
\rho(k) \triangleq \min_{\mathrm{partition}\ \seta_1,\ldots, \seta_k} \max_{1\leq i\leq k} \cond_\graphg(\seta_i).
\]

Next we define the matrices of any $\graphg=(\vertexset_\graphg, \edgeset_\graphg, \weight)$. Let $\degm_\graphg \in \R^{n \times n}$ be the diagonal matrix defined by $(\degm_\graphg)_{uu} = \deg_\graphg(u)$ for all $u \in \vertexset_\graphg$, and we  denote by  $\adj_\graphg \in \R^{n\times n}$  the \emph{adjacency matrix}  of $\graphg$, where $(\adj_\graphg)_{uv} = \weight_{uv}$ for all $u, v \in \vertexset_\graphg$.
The \emph{normalised Laplacian matrix} of $\graphg$ is defined by $\lapn_\graphg \triangleq \identity - \degmhalfneg_\graphg \adj_\graphg \degmhalfneg_\graphg$, where $\identity$ is the $n \times n$ identity matrix.  Since  $\lapn_\graphg$ is symmetric and real-valued, it has $n$ real eigenvalues denoted by $\lambda_1 \leq \ldots \leq \lambda_n$; we use $\vecf_i \in \R^n$  to denote the eigenvectors corresponding to $\lambda_i$ for any $1\leq i\leq n$. It is known that $\lambda_1 = 0$ and $\lambda_n \leq 2$~\cite{chungSpectralGraphTheory1997}.

For any sets $\sets$ and $\sett$, the symmetric difference between $\sets$ and $\sett$ is defined by \[\sets \triangle \sett = (\sets \setminus \sett) \union (\sett \setminus \sets).\] For any $k\in \Z^+$, we define $[k]\triangleq \{1,\ldots, k\}$. We sometimes drop the subscript $\graphg$ when it is clear from the context. 
The following higher-order Cheeger inequality will be used in our analysis.

\begin{lemma}[\cite{leeMultiwaySpectralPartitioning2014}]\label{lem:high-order-cheeger} It holds for any $k\in[n]$ that
\[
\lambda_k/2 \leq \rho(k) \leq O\left(k^3\right)\sqrt{\lambda_{k}}.\]
\end{lemma}

\section{Encoding the Cluster-Structure into the Eigenvectors of \texorpdfstring{$\lapn_\graphg$}{N} \label{sec:structure}}
Let $\{\sets_i\}_{i=1}^k$ be any optimal $k$-way partition that achieves $\rho(k)$.  We define the indicator vector of cluster $\sets_i$ by 
\begin{equation}
\label{eq:const_g}
\indicatorvec_i(u)\triangleq 
\left\{
	\begin{array}{ll}
		1  & \mbox{if } u \in \sets_i, \\
		0 & \mbox{otherwise},
	\end{array}
\right.
\end{equation}
and the corresponding \emph{normalised indicator vector} by 
\[
\barg_i \triangleq \frac{\degmhalf \indicatorvec_i}{ \|\degmhalf \indicatorvec_i \|}.  
\]
One of the basic results in spectral graph theory states that $\graphg$ consists of at least $k$ connected components if and only if $\lambda_i=0$ for any $i \in [k]$, and
$\mathrm{span}\left(\{\vecf_i\}_{i=1}^k\right) = \mathrm{span}\left(\{\barg_i \}_{i=1}^k\right)$~\cite{chungSpectralGraphTheory1997}.
Hence,  one would expect that, when $\graphg$ consists of $k$ densely connected components~(clusters) connected by sparse cuts, the bottom eigenvectors $\{\vecf_i\}_{i=1}^k$ of $\lapn_\graphg$ are close to $\{\barg_i\}_{i=1}^k$. This intuition explains the practical success of spectral methods for graph clustering, and forms the basis of many theoretical studies on various
spectral clustering algorithms~(e.g.,~\cite{kwokImprovedCheegerInequality2013a, leeMultiwaySpectralPartitioning2014, ngSpectralClusteringAnalysis2001, vonluxburgTutorialSpectralClustering2007}). 

Turning this intuition into a mathematical statement, Peng et al.~\cite{pengPartitioningWellClusteredGraphs2017} study the quantitative relationship   between $\{\vecf_i\}_{i=1}^k$ and $\{\barg_i\}_{i=1}^k$ through the function $\Upsilon(k)$ defined by
\begin{equation}\label{eq:defineupsilon}
    \Upsilon(k) \triangleq \frac{\lambda_{k+1}}{\rho(k)}.
\end{equation}
To explain   the meaning of $\Upsilon(k)$, we  assume that $\graphg$ has $k$ well-defined clusters $\{\sets_i\}_{i=1}^k$. By definition, the values of $\cond(\sets_i)$ for every $\sets_i$, as well as    $\rho(k)$, are low;
on the other hand, any $(k+1)$-way partition of $\vertexset_\graphg$ would separate the vertices of  some     $\sets_i$, and as such $\rho(k+1)$'s value will be  high.
Combining this with the higher-order Cheeger inequality, some lower bound on  $\Upsilon(k)$ would be sufficient to ensure that $\graphg$ has exactly $k$ clusters. 
In their work, Peng et~al.~\cite{pengPartitioningWellClusteredGraphs2017} assumes   $\Upsilon (k) =\Omega(k^2)$, and    proves that the space spanned by $\{\vecf_i\}_{i=1}^k$ and the one spanned by $\{\barg_i\}_{i=1}^k$ are close to each other. 
Specifically, they show that
\begin{enumerate}
\item  every $\barg_i $ is close to some linear combination of  $\{\vecf_i\}^k_{i=1}$, denoted by $\hatf_i$, i.e.,  it holds that \[\|\barg_i - \hatf_i\|^2\leq 1/\Upsilon(k);\]
    \item  every   $\vecf_i$ is close to some linear combination of $\{
\barg_i\}_{i=1}^k$, denoted by $\hatg_i$, i.e., it holds that \[\|\vecf_i - \hatg_i \|^2\leq 1.1 k/\Upsilon(k).\]
\end{enumerate}
In essence, their so-called structure theorem gives a quantitative explanation on why spectral methods work for graph clustering
when there is a clear cluster-structure in $\graphg$ characterised by $\Upsilon(k)$.
As it holds for  graphs with clusters of different sizes and edge densities, this structure theorem has been shown to be a powerful tool in analysing clustering algorithms, and inspired many subsequent works~(e.g.,~\cite{chenCommunicationoptimalDistributedClustering2016, czumajTestingClusterStructure2015, kloumannBlockModelsPersonalized2017, kolevNoteSpectralClustering2016, louisPlantedModelsKway2019, mizutaniImprovedAnalysisSpectral2021,pengRobustClusteringOracle2020,pengAverageSensitivitySpectral2020,sunDistributedGraphClustering2019}).  

In this section we show that a stronger statement of the original structure theorem holds under a much weaker assumption. Our result is summarised as follows:

\begin{theorem}[The Stronger Structure Theorem] \label{thm:struc1}
The following statements hold:
\begin{enumerate}
    \item For any $i\in[k]$, there is $\hatf_i\in\mathbb{R}^n$, which is  
    a linear combination of $\vecf_1,\ldots, \vecf_k$, such that \[\|\barg_i - \hatf_i\|^2 \leq 1/\Upsilon(k).\]
    \item There are vectors $\hatg_1,\ldots, \hatg_k$, each of which is a linear combination of $\barg_1,\ldots, \barg_k$, such that \[\sum_{i = 1}^k \norm{\vecf_i - \hatg_i}^2 \leq k /\Upsilon(k).\]
\end{enumerate}
\end{theorem}
\begin{proof}
Let $\hatf_i = \sum_{j=1}^k \langle \barg_i, \vecf_j\rangle \vecf_j $, and we write $\barg_i$ as a linear combination of the vectors $\vecf_1, \ldots, \vecf_n$ by 
   $
        \barg_i = \sum_{j = 1}^n \langle \barg_i, \vecf_j \rangle \vecf_j$.
   Since $\hatf_i$ is a projection of $\barg_i$, we have that $\barg_i - \hatf_i$ is perpendicular to $\hatf_i$ and 
    \begin{align*}
        \norm{\barg_i - \hatf_i}^2 & = \norm{\barg_i}^2 - \norm{\hatf_i}^2 \\
        & = \left(\sum_{j = 1}^n \langle \barg_i, \vecf_j \rangle^2 \right) - \left(\sum_{j = 1}^{k} \langle \barg_i, \vecf_j \rangle^2 \right)  \\
        & = \sum_{j = k + 1}^n \langle \barg_i, \vecf_j \rangle^2.
    \end{align*}
     Now, let us consider the quadratic form
    \begin{align}
        \barg_i^\transpose \mathcal{N}_G \barg_i & = \left(\sum_{j = 1}^n \langle \barg_i, \vecf_j \rangle \vecf_j^\transpose \right) \mathcal{L}_G \left(\sum_{j = 1}^n \langle \barg_i, \vecf_j \rangle \vecf_j\right) \nonumber \\
        & = \sum_{j = 1}^n \langle \barg_i, \vecf_j \rangle^2 \lambda_j \nonumber  \\
        & \geq \lambda_{k + 1} \norm{\barg_i - \hatf_i}^2, \label{eq:lbquad}
    \end{align}
    where the last inequality follows by the fact that $\lambda_i\geq 0$ holds for any $1\leq i\leq n$. This gives us that  
    \begin{align}
        \barg_i^\transpose \lapn_\graphg \barg_i & = \sum_{(u, v) \in \edgeset_\graphg} \weight(u, v) \left(\frac{\barg_i(u)}{\sqrt{\deg(u)}} - \frac{\barg_i(v)}{\sqrt{\deg(v)}}\right)^2 \nonumber \\
        & = \sum_{(u, v) \in \edgeset_\graphg} \weight(u, v) \left(\frac{\indicatorvec_i(u)}{\sqrt{\vol(\sets_i)}} - \frac{\indicatorvec_i(v)}{\sqrt{\vol(\sets_i)}}\right)^2 \nonumber \\
        & = \frac{\weight(\sets_i, \vertexset \setminus \sets_i)}{\vol(\sets_j)}\nonumber\\
        & \leq \rho(k).  \label{eq:upquad}
    \end{align}
    Combining \eqref{eq:lbquad} with \eqref{eq:upquad}, we have that
    \[
    \norm{\barg_i - \hatf_i}^2 \leq \frac{\barg_i^\transpose \lapn_\graphg \barg_i}{\lambda_{k+1}} \leq \frac{\rho(k)}{\lambda_{k+1}} \leq \frac{1}{\Upsilon (k)},
    \]
    which proves the first statement of the theorem.

Now we prove the second statement.
We define for any $1\leq i \leq k$ that 
 $\hatg_i = \sum_{j=1}^k \langle \vecf_i, \barg_j\rangle \barg_j$, and  have that 
    \begin{align*}
        \sum_{i = 1}^k \norm{\vecf_i - \hatg_i}^2
        & = \sum_{i = 1}^k \left( \norm{\vecf_i}^2 - \norm{\hatg_i}^2 \right) \\
        & = k - \sum_{i = 1}^k \sum_{j = 1}^k \langle \barg_j, \vecf_i \rangle^2   \\
        & = \sum_{j = 1}^k \left( 1 - \sum_{i = 1}^k \langle \barg_j, \vecf_i \rangle^2 \right) \\
        & = \sum_{j = 1}^k \left( \norm{\barg_j}^2 - \norm{\hatf_j}^2 \right) \\
        & = \sum_{j = 1}^k \norm{\barg_j - \hatf_j}^2\\
        & \leq \sum_{j = 1}^k \frac{1}{\Upsilon(k)}\\
        & = \frac{k}{\Upsilon(k)},
    \end{align*}
    where the last inequality follows by the first statement of  Theorem~\ref{thm:struc1}.
\end{proof}

To examine the significance of Theorem~\ref{thm:struc1}, we first highlight  that these two statements hold for any $\Upsilon(k)$, while the original structure theorem relies on the assumption that $\Upsilon(k)=\Omega(k^2)$. Since  $\Upsilon(k)=\Omega(k^2)$ is a strong and even questionable assumption when $k$ is large, e.g., $k=\bigomega{\polylog(n)}$, obtaining these statements for general $\Upsilon(k)$ is important.
Secondly, our second statement of Theorem~\ref{thm:struc1} significantly improves the original theorem. Specifically, instead of stating  $\|\vecf_i - \hatg_i \|^2\leq 1.1 k/\Upsilon(k)$ for any $i\in [k]$, our second statement shows that $\sum_{i = 1}^k \norm{\vecf_i - \hatg_i}^2 \leq k /\Upsilon(k)$;
hence, it holds in expectation that $\norm{\vecf_i - \hatg_i}^2 \leq 1 /\Upsilon(k)$,  the upper bound of which matches the first statement.
This implies that the vectors $\vecf_1,\ldots, \vecf_k$ and $\barg_1,\ldots, \barg_k$ can be linearly approximated by each other  with  \emph{roughly the same} approximation guarantee. 
Thirdly, rather than employing the machinery from matrix analysis used by Peng et al.~\cite{pengPartitioningWellClusteredGraphs2017}, 
to prove the original theorem, our  proof is simple and purely linear-algebraic.
Therefore, we believe that both of our stronger statements and much simplified proof are significant, and could have further applications in graph clustering and related problems. 

\section{Tighter Analysis of Spectral Clustering\label{sec:analysis1}}
In this section, we analyse the spectral clustering algorithm.
For any input graph $\graphg=(\vertexset_\graphg, \edgeset_\graphg)$ and  $k\in[n]$, spectral clustering consists of the three steps below: 
\begin{enumerate}
    \item compute the eigenvectors $\vecf_1,\ldots \vecf_k$ of $\lapn_\graphg$, and embed each $u\in \vertexset_\graphg$ to the point $F(u) \in \R^k$ according to 
    \begin{equation}\label{eq:embedding}
     F(u) \triangleq \frac{1}{\sqrt{\deg(u)}} \left( \vecf_1(u),\ldots, \vecf_k(u)\right)^\transpose;
    \end{equation}
    \item apply $k$-means on the embedded points $\{ F(u)\}_{u\in \vertexset_\graphg}$;
    \item partition $\vertexset_\graphg$ into $k$ clusters based on  the output of  $k$-means.
\end{enumerate}

We will consider spectral clustering for graphs with clusters of \emph{almost-balanced} size.
\begin{definition} \label{def:almostBalanced}
    Let $\graphg$ be a graph with $k$ clusters $\{\sets_i\}_{i=1}^k$.
    We say that the clusters are almost-balanced if
    $(1/2) \cdot \vol(\vertexsetg)/k \leq \vol(\sets_i) \leq 2 \cdot \vol(\vertexsetg)/k$ 
    for all $i \in \{1, \ldots, k\}$.
\end{definition}
Our main result is given in Theorem~\ref{thm:sc_guarantee}, where we take $\APT$ to be the approximation ratio of the $k$-means algorithm used in spectral clustering.
Recall that we can take $\APT$ to be some small constant~\cite{kumarSimpleLinearTime2004}.

\begin{theorem}\label{thm:sc_guarantee}
Let $\graphg$ be a graph with $k$ clusters $\{\sets_i\}_{i=1}^k$ of almost balanced size, and
$\Upsilon(k) \geq 2176 (1 + \APT)$.
Let $\{\seta_i\}_{i=1}^k$ be the output of spectral clustering and,  without loss of generality,  the  optimal correspondent of $\seta_i$ is $\sets_i$. Then, it holds that
\[
        \sum_{i = 1}^k \vol\left(\seta_i \triangle \sets_i \right) \leq 2176 ~(1 + \APT)~ \frac{\vol(\vertexsetg)}{\Upsilon(k)}.
\] 
\end{theorem}
Notice that some condition on $\Upsilon(k)$  is needed 
to ensure that an input graph $\graphg$ has $k$ well-defined clusters, so that misclassified vertices can be formally defined.
Taking this into account, the most significant feature of Theorem~\ref{thm:sc_guarantee} is its upper bound of misclassified vertices  with respect to $\Upsilon(k)$: our result holds, and is non-trivial, as long as 
  $\Upsilon(k)$ is lower bounded by  some constant\footnote{Note that we can take any constant approximation in Definition~\ref{def:almostBalanced} with a different corresponding constant in Theorem~\ref{thm:sc_guarantee}.}.
This significantly improves most of the previous results of graph clustering algorithms,  which make stronger assumptions on the input graphs. For example, Peng et al.~\cite{pengPartitioningWellClusteredGraphs2017} assumes that $\Upsilon(k)=\bigomega{k^3}$, Mizutani~\cite{mizutaniImprovedAnalysisSpectral2021} assumes that $\Upsilon(k) =\bigomega{k}$, the algorithm presented in Gharan and Trevisan~\cite{gharanPartitioningExpanders2014} assumes that $\lambda_{k+1} = \bigomega{\mathrm{poly}(k)\lambda^{1/4}_k}$, and the one presented in Dey et al.~\cite{deySpectralConcentrationGreedy2019} further assumes some condition with respect to  $k$, $\lambda_k$, and the maximum degree of $\graphg$.
While  these assumptions require at least a linear dependency on $k$, making it difficult  for the  instances with a large value of $k$ to satisfy,  our result suggests  that the performance of spectral clustering can be rigorously  analysed for these graphs.
In particular, compared with previous work, our result better justifies  the widely used eigen-gap heuristic for spectral clustering~\cite{ngSpectralClusteringAnalysis2001,vonluxburgTutorialSpectralClustering2007}.
This heuristic suggests that spectral clustering works when the value of $|\lambda_{k+1} -\lambda_k|$ is much larger than $|\lambda_{k} -\lambda_{k-1}|$, and in practice, the ratio between the two gaps  is usually a constant rather than some function of $k$.
 
\subsection{Properties of Spectral Embedding}
Now we  study  the properties  of the spectral embedding defined in \eqref{eq:embedding}, and show in the next subsection how to use these properties to prove Theorem~\ref{thm:sc_guarantee}.
For every cluster $\sets_i$, we define the vector $\peye \in \R^k$ by 
\[
\peye(j) = \frac{1}{\sqrt{\vol(\sets_i) }} \langle \vecf_j, \barg_i \rangle,
\]
and view these $\{\peye\}_{i=1}^k$ as the approximate centres of the embedded points from the optimal clusters $\{\sets_i\}_{i=1}^k$.
We prove that the total $k$-means cost of the embedded points can be upper bounded as follows: 

\begin{lemma} \label{lem:total_cost}
    It holds that
    \[
        \sum_{i = 1}^k \sum_{u \in \sets_i} \deg(u) \norm{F(u) - \peye}^2 \leq \frac{k}{\Upsilon(k)}.
    \]
\end{lemma}
\begin{proof}
    We have
    \begin{align*}
        \sum_{i = 1}^k \sum_{u \in \sets_i} \deg(u) \norm{F(u) - \peye}^2 & = \sum_{i = 1}^k \sum_{u \in \sets_i} \deg(u) \left[ \sum_{j = 1}^k \left(\frac{\vecf_j(u)}{\sqrt{\deg(u)}} - \frac{\langle \barg_i, \vecf_j\rangle}{\sqrt{\vol(\sets_i)}}\right)^2 \right] \\
        & = \sum_{i = 1}^k \sum_{u \in \sets_i} \sum_{j = 1}^k \left(\vecf_j(u) - \langle \barg_i, \vecf_j\rangle \barg_i(u)\right)^2 \\
        & = \sum_{i = 1}^k \sum_{u \in \sets_i} \sum_{j = 1}^k \left(\vecf_j(u) - \hatg_j(u)\right)^2 \\
        & = \sum_{j = 1}^k \norm{\vecf_j - \hatg_j}^2 \\
        & \leq \frac{k}{\Upsilon(k)},
    \end{align*}
    where the final inequality follows by the second statement of Theorem~\ref{thm:struc1} and it holds for  $u \in \sets_x$ that  $
        \hatg_i(u) = \sum_{j = 1}^k \langle \vecf_i, \barg_j\rangle \barg_j(u) = \langle \vecf_i, \barg_x\rangle \barg_x(u)$.  
\end{proof}
The importance of Lemma~\ref{lem:total_cost} is that, although the optimal centres for $k$-means are unknown, the existence of $\{\peye\}_{i=1}^k$ is sufficient to  show  that the cost of an optimal $k$-means clustering on $\{F(u)\}_{u\in \vertexsetg}$ is at most $k/\Upsilon(k)$.
Since one can always use an $O(1)$-approximate $k$-means algorithm for spectral clustering~(e.g.,~\cite{kanungoLocalSearchApproximation2004,kumarSimpleLinearTime2004}), the cost of the output of $k$-means on $\{F(u)\}_{u\in \vertexsetg}$ is $O\left(k/\Upsilon(k)\right)$.
Next, we show that the
length of $\peye$ is approximately equal to
$1/\vol(\sets_i)$,
which will be useful in our later analysis.

\begin{lemma} \label{lem:pnorm}
    It holds for  any $i \in [k]$ that 
    \[
        \frac{1}{\vol(\sets_i)} \left(1 - \frac{1}{\Upsilon(k)}\right) \leq \norm{\vecp^{(i)}}^2 \leq \frac{1}{\vol(\sets_i)}.
    \]
\end{lemma}

\begin{proof} 
    By definition, we have
    \begin{align*}
        \vol(\sets_i) \norm{\vecp^{(i)}}^2 & = \sum_{j = 1}^k \langle \vecf_j, \barg_i \rangle^2  = \norm{\hatf_i}^2  = 1 - \norm{\hatf_i - \barg_i}^2  \geq 1 - \frac{1}{\Upsilon(k)}, 
    \end{align*}
    where the inequality follows by Theorem~\ref{thm:struc1}. The other direction of the inequality follows similarly.
\end{proof}

In the remainder of this subsection, we will prove a sequence of lemmas showing that any pair of $\peye$ and $\pj$ are well separated.
Moreover, notice that their distance is essentially independent of $k$ and $\Upsilon(k)$, as long as $\Upsilon(k) \geq 20$.

\begin{lemma} \label{lem:normp_diff}
It holds for any different $i,j\in[k]$ that 
\[
    \left\|\sqrt{\vol(\sets_i)}\cdot \vecp^{(i)} - \sqrt{\vol(\sets_j)}\cdot \vecp^{(j)}\right\|^2 \geq 2 - \frac{8}{\Upsilon(k)}.
\]
\end{lemma}

\begin{proof}
    We have
    \begin{align*}
   \lefteqn{\left\|\sqrt{\vol(\sets_i)}\cdot  \vecp^{(i)} - \sqrt{\vol(\sets_j)}\cdot  \vecp^{(j)}\right\|^2}\\
   & = \sum_{x = 1}^k \left(\langle \vecf_x, \barg_i \rangle - \langle \vecf_x, \barg_j \rangle \right)^2 \\
    & = \left( \sum_{x = 1}^k \langle \vecf_x, \barg_i\rangle^2 \right) + \left(\sum_{x = 1}^k \langle \vecf_x, \barg_j \rangle^2 \right) - 2 \sum_{x = 1}^k \langle \vecf_x, \barg_i\rangle \langle \vecf_x, \barg_j \rangle \\
    & \geq \norm{\hatf_i}^2 + \norm{\hatf_j}^2 - 2 \abs{\hatf_i^\transpose \hatf_j} \\
    & \geq 2 \left(1 - \frac{1}{\Upsilon(k)}\right) - 2 \abs{(\barg_i + \hatf_i - \barg_i)^\transpose (\barg_j + \hatf_j - \barg_j)} \\
    & = 2 \left(1 - \frac{1}{\Upsilon(k)}\right) - 2 \left( \abs{\langle \barg_i, \hatf_j - \barg_j \rangle + \langle \barg_j, \hatf_i - \barg_i \rangle + \langle \hatf_i - \barg_i, \hatf_j - \barg_j\rangle}  \right) \\
    & \geq 2 \left(1 - \frac{1}{\Upsilon(k)}\right) - 6 \cdot \frac{1}{\Upsilon(k)}\\
    & \geq 2 - \frac{8}{\Upsilon(k)}. \qedhere
    \end{align*}
\end{proof}

\begin{lemma} \label{lem:normp_diff2}
It holds for any different $i,j\in[k]$ that  
\[
        \left\|\frac{\vecp^{(i)}}{\norm{\peye}} - \frac{\pj}{\norm{\pj}}\right\|^2 \geq 2 - \frac{20}{\Upsilon(k)}.
    \]
\end{lemma} 

\begin{proof}
    Assume without loss of generality that $$\sqrt{\vol(\sets_i)} \norm{\peye} \leq \sqrt{\vol(\sets_j)} \norm{\pj}.$$
    Let $\veca_i = \sqrt{\vol(\sets_i)} \peye$ and $\veca_j = \sqrt{\vol(\sets_j)} \pj$ and notice that $\norm{\veca_i} \leq \norm{\veca_j} \leq 1$.
    Then, using Lemma~\ref{lem:pnorm} we have \begin{align*}
        \left\|\frac{\vecp^{(i)}}{\norm{\peye}} - \frac{\pj}{\norm{\pj}}\right\| & \geq \left\|\veca_i - \frac{\norm{\veca_i}}{\norm{\veca_j}} \veca_j\right\| \\
        & \geq \norm{\veca_i - \veca_j} - \left(\norm{\veca_j} - \norm{\veca_i}\right) \\
        & \geq \sqrt{2 - \frac{8}{\Upsilon(k)}}  - \left(\sqrt{\vol(\sets_j)}\cdot \left\|\pj\right\| - \sqrt{\vol(\sets_i)} \left\|\peye\right\|\right) \\
        & \geq \sqrt{2}\left(1 - \frac{4}{\Upsilon(k)}\right) + \sqrt{1 - \frac{1}{\Upsilon(k)}} - 1 \\
        & \geq \sqrt{2} - \frac{4\sqrt{2}}{\Upsilon(k)} - \frac{1}{\Upsilon(k)}\\
        &\geq \sqrt{2} - \frac{7}{\Upsilon(k)},
    \end{align*}
    where the second inequality follows by the triangle inequality, and the third and fourth use Lemma~\ref{lem:normp_diff}. We also use the fact that for $x \leq 1$, it is the case that $\sqrt{1 - x} \geq 1 - x$.
    This gives
    \begin{align*}
        \left\|\frac{\vecp^{(i)}}{\norm{\peye}} - \frac{\pj}{\norm{\pj}}\right\|^2 & \geq 2 - \frac{14 \sqrt{2}}{\Upsilon(k)}  \geq 2 - \frac{20}{\Upsilon(k)},
    \end{align*}
    which proves the lemma.
\end{proof}

\begin{lemma} \label{lem:p_diff}
    It holds for  any $i, j \in [k]$ with $i \neq j$ that 
    \[
        \norm{\peye - \pj}^2 \geq \frac{1}{\min\{\vol(\sets_i), \vol(\sets_j)\}} \left(\frac{1}{2} - \frac{8}{\Upsilon(k)} \right). 
    \]
\end{lemma}
\begin{proof}
    Assume without loss of generality that $\norm{\peye} \geq \norm{\pj}$. Then, let $\norm{\pj} = \alpha \norm{\peye}$ for some $\alpha \in [0, 1]$.
    By Lemma~\ref{lem:pnorm}, it holds that
    \[
        \left\|\peye\right\|^2 \geq \frac{1}{\min\{\vol(\sets_i), \vol(\sets_j)\|} \left(1 - \frac{1}{\Upsilon(k)}\right).
    \]
    Additionally, notice that by the proof of Lemma~\ref{lem:normp_diff2},  
    \[
        \left\langle \frac{\peye}{\norm{\peye}}, \frac{\pj}{\norm{\pj}} \right\rangle \leq \sqrt{2} - \frac{1}{2} \norm{\frac{\peye}{\norm{\peye}} - \frac{\pj}{\norm{\pj}}} \leq \frac{\sqrt{2}}{2} + \frac{7}{2 \Upsilon(k)},
    \]
    where we use the fact that if $x^2 + y^2 = 1$, then $x + y \leq \sqrt{2}$.
    One can understand the equation above by considering the right-angled triangle with one edge given by $\peye / \norm{\peye}$ and another edge given by $(\peye / \norm{\peye}) . (\pj / \norm{\pj})$.
    Then,
    \begin{align*}
        \norm{\peye - \pj}^2 & = \norm{\peye}^2 + \norm{\pj}^2 - 2 \left\langle \frac{\peye}{\norm{\peye}}, \frac{\pj}{\norm{\pj}} \right\rangle \norm{\peye} \norm{\pj} \\
        & \geq (1 + \alpha) \norm{\peye}^2 - \left(\sqrt{2} + \frac{7}{\Upsilon(k)}\right) \alpha \norm{\peye}^2 \\
        & \geq \left(1 - (\sqrt{2} - 1) \alpha  - \frac{7}{\Upsilon(k)}\right) \norm{\peye}^2 \\
        & \geq  \frac{1}{\min\{\vol(\sets_i), \vol(\sets_j)\}}  \left(\frac{1}{2} - \frac{7}{\Upsilon(k)} \right) \left(1 - \frac{1}{\Upsilon(k)}\right) \\
        & \geq  \frac{1}{\min\{\vol(\sets_i), \vol(\sets_j)\}} \left(\frac{1}{2} - \frac{8}{\Upsilon(k)}\right),
    \end{align*}
    which completes the proof.
\end{proof}

We remark that, despite the similarity in their formulation, the technical lemmas presented in this subsection are stronger than the ones in Peng et al.~\cite{pengPartitioningWellClusteredGraphs2017}.  These results are obtained through our stronger structure theorem~(Theorem~\ref{thm:struc1}), and are crucial for us to prove Theorem~\ref{thm:sc_guarantee}.

\subsection{Proof of Theorem~\ref{thm:sc_guarantee}
\label{sec:sc_proof_sketch}}
In this subsection, we prove
Theorem~\ref{thm:sc_guarantee},
and explain why a mild condition like $\Upsilon(k) = \bigomega{1}$ suffices for spectral clustering to perform well in practice. 
Let   $\{\seta_i\}_{i=1}^k$ be the output of spectral clustering,  and we denote  the centre of the embedded points $\{F(u)\}$ for any $\seta_i$  by $\vecc_i$.
As the starting point of our analysis, we claim that every $\vecc_i$ will be close to its ``optimal'' correspondent $\vecp^{(\sigma(i))}$ for some $\sigma(i)\in[k]$. That is, the actual centre of embedded points from every $\seta_i$ is close to the approximate centre of the embedded points from some optimal $\sets_i$.
To formalise this,  we define the function $\sigma:[k]\rightarrow[k]$ by
\begin{equation}\label{eq:defsigma}
    \sigma(i) = \argmin_{j \in [k]} \norm{\pj - \vecc_{i}};
\end{equation}
that is, cluster $\seta_i$ should correspond to $\sets_{\sigma(i)}$ in which the value of $\|\vecp^{(\sigma(i))} - \vecc_i\|$ is the lowest among all the distances between $\vecc_i$ and all of the $\vecp^{(j)}$ for $j\in [k]$.
However, one needs to be cautious as \eqref{eq:defsigma} wouldn't necessarily define   a permutation, and there might exist different $i,i'\in[k]$ such that both of $\seta_i$ and $\seta_{i'}$ map to the same $\sets_{\sigma(i)}$. 
Taking this into account, for any fixed $\sigma:[k]\rightarrow[k]$ and $i\in[k]$, we  further  define $\setm_{\sigma, i}$ by
\begin{equation}\label{eq:defmset}
    \setm_{\sigma, i} \triangleq \bigcup_{j : \sigma(j) = i} \seta_j.
\end{equation}
The following lemma shows that, when mapping every output  $\seta_i$ to $\sets_{\sigma(i)}$, the total ratio of misclassified volume with respect to each cluster can be upper bounded:
\begin{lemma} \label{lem:cost_lower_bound}
Let $\{\seta_i\}_{i=1}^k$ be the output of spectral clustering, and $\sigma$ and  $\setm_{\sigma,i}$ be defined as  in \eqref{eq:defsigma} and \eqref{eq:defmset}.
If $\Upsilon(k) \geq 32$, then
\[
        \sum_{i = 1}^k \frac{\vol(\setm_{\sigma, i} \triangle \sets_i)}{\vol(\sets_i)} \leq 64 (1 + \APT) \frac{k}{\Upsilon(k)}.
    \]
 \end{lemma}
 \begin{proof}
    Let us define $\setb_{ij} = \seta_i \cap \sets_j$ to be the vertices in $\seta_i$ which belong to the true cluster $\sets_j$.
    Then, we have that
    \begin{align}
        \sum_{i = 1}^k \frac{\vol(\setm_{\sigma, i} \triangle \sets_i)}{\vol(\sets_i)} & = \sum_{i = 1}^k \sum_{\substack{j = 1 \\ j \neq \sigma(i)}}^k \vol(\setb_{ij}) \left(\frac{1}{\vol(\sets_{\sigma(i)})} + \frac{1}{\vol(\sets_j)} \right) \nonumber \\
        & \leq 2 \sum_{i = 1}^k \sum_{\substack{j = 1 \\ j \neq \sigma(i)}}^k \frac{\vol(\setb_{ij})}{\min\{\vol(\sets_{\sigma(i)}), \vol(\sets_j)\}}, \label{eq:up_sym_ratio}
    \end{align}
and that 
\begin{align*}
        \mathrm{COST}(\seta_1, \ldots \seta_k) & = \sum_{i = 1}^k \sum_{u \in \seta_i} \deg(u) \norm{F(u) - \vecc_i}^2 \\
        & \geq \sum_{i = 1}^k \sum_{\substack{1\leq j \leq  k \\ j \neq \sigma(i)}}   \sum_{u \in \setb_{ij}} \deg(u) \norm{F(u) - \vecc_i}^2 \\
        & \geq \sum_{i = 1}^k \sum_{\substack{1\leq j \leq k  \\ j \neq \sigma(i)}}  \sum_{u \in \setb_{ij}} \deg(u) \left(\frac{\norm{\pj - \vecc_i}^2}{2} - \norm{\pj - F(u)}^2 \right) \\
        & \geq \sum_{i = 1}^k \sum_{\substack{1\leq j \leq k \\ j \neq \sigma(i)}}  \sum_{u \in \setb_{ij}} \frac{\deg(u) \norm{\pj - \vecp^{(\sigma(i))}}^2}{8} - \sum_{i = 1}^k \sum_{\substack{1\leq j \leq k \\ j \neq i}}  \sum_{u \in \setb_{ij}} \deg(u) \norm{\pj - F(u)}^2 \\
        & \geq \sum_{i = 1}^k \sum_{\substack{1\leq j\leq k \\ j \neq \sigma(i)}}  \vol(\setb_{ij}) \frac{\norm{\pj - \vecp^{(\sigma(i))}}^2}{8} - \sum_{i = 1}^k \sum_{u \in \sets_i} \deg(u) \norm{\peye - F(u)}^2 \\
        & \geq \sum_{i = 1}^k \sum_{\substack{1\leq j\leq k \\ j \neq \sigma(i)}}   \frac{\vol(\setb_{ij})}{8 \min\{\vol(\sets_{\sigma(i)}), \vol(\sets_j)\}}\left(\frac{1}{2} - \frac{8}{\Upsilon(k)}\right) - \sum_{i = 1}^k \sum_{u \in \sets_i} \deg(u) \norm{\peye - F(u)}^2 \\
        & \geq \frac{1}{16} \cdot\left( \sum_{i = 1}^k \frac{\vol(\setm_{\sigma, i} \triangle \sets_i)}{\vol(\sets_i)}\right) \left(\frac{1}{2} - \frac{8}{\Upsilon(k)}\right) - \frac{k}{\Upsilon(k)},
    \end{align*}
    where the second inequality follows by the inequality of  \[\norm{\veca - \vecb}^2 \geq \frac{\norm{\vecb - \vecc}^2}{2} - \norm{\veca - \vecc}^2,\] the third inequality follows since $\vecc_i$ is closer to $\vecp^{(\sigma(i))}$ than $\vecp^{(j)}$, the fifth one follows from Lemma~\ref{lem:p_diff}, and the last one follows by \eqref{eq:up_sym_ratio}.
    
   On the other side, since $\mathrm{COST}(\seta_1,\ldots, \seta_k) \leq \APT\cdot k/\Upsilon(k)$, we have that 
    \[
    \frac{1}{16} \cdot\left( \sum_{i = 1}^k \frac{\vol(\setm_{\sigma, i} \triangle \sets_i)}{\vol(\sets_i)}\right) \left(\frac{1}{2} - \frac{8}{\Upsilon(k)}\right) - \frac{k}{\Upsilon(k)} \leq \APT\cdot\frac{k}{\Upsilon(k)}.
    \] This implies that 
    \begin{align*}
            \sum_{i = 1}^k \frac{\vol(\setm_{\sigma, i} \triangle \sets_i)}{\vol(\sets_i)}   & \leq 16\cdot   (1+\APT)\cdot \frac{k}{\Upsilon(k)}\cdot  \left(\frac{1}{2} - \frac{8}{\Upsilon(k)}\right)^{-1}\\
            &\leq 64\cdot   (1+\APT)\cdot \frac{k}{\Upsilon(k)},
    \end{align*}
    where the last inequality holds by the assumption that $\Upsilon(k)\geq 32$. 
\end{proof}
 
 It remains to  study the case in which $\sigma$ isn't a permutation. Notice that, if this occurs,  there is some $i\in[k]$ such that $\setm_{\sigma,i}=\emptyset$, and different values of $x,y\in[k]$ such that $\sigma(x)=\sigma(y)=j$ for some $j\neq i$.
 Based on this, we construct the function $\sigma':[k]\rightarrow [k]$ from $\sigma$ based on the following procedure: 
\begin{itemize}
    \item Set $\sigma'(z)=i$ if $z=x$;
    \item Set $\sigma'(z)=\sigma(z)$ for any other $z\in [k]\setminus \{x\}$.
\end{itemize}
Notice that one can construct $\sigma'$ in this way as long as $\sigma$ isn't a permutation, and this constructed $\sigma'$ reduces the number of $\setm_{\sigma,i}$ being $\emptyset$ by one.
We show one only needs to construct such $\sigma'$ at most $\bigo{k/\Upsilon(k)}$ times
to obtain the final permutation called $\sigma^{\star}$, and it holds for $\sigma^{\star}$ that 
\[
        \sum_{i = 1}^k \frac{\vol(\setm_{\sigma^{\star}, i} \triangle \sets_i)}{\vol(\sets_i)} =\bigo{\frac{k}{\Upsilon(k)}}.
    \]
Combining this with the fact that $\vol(\sets_i) = \Theta(\vol(\vertexsetg)/k)$ for any $i\in[k]$ proves  Theorem~\ref{thm:sc_guarantee}.

\begin{proof}[Proof of Theorem~\ref{thm:sc_guarantee}]
\newcommand{\voldiffi}{\vol(\setm_{\sigma, i} \triangle \sets_i)}
\newcommand{\voldiffip}{\vol(\setm_{\sigma', i} \triangle \sets_i)}
\newcommand{\voldiffjp}{\vol(\setm_{\sigma', j} \triangle \sets_j)}
\newcommand{\voldiffj}{\vol(\setm_{\sigma, j} \triangle \sets_j)}
    By Lemma~\ref{lem:total_cost}, we have
    \[
        \mathsf{COST}(\sets_1, \ldots, \sets_k) \leq  \frac{k}{\Upsilon(k)}.
        \]
    Combining this with the fact that one can apply an approximate $k$-means clustering algorithm with approximation ratio $\APT$ for spectral clustering,  we have that
    \[
        \mathsf{COST}(\seta_1, \ldots, \seta_k) \leq \APT\cdot\frac{k}{\Upsilon(k)}.
        \]
    Then, let $\sigma: [1, k] \rightarrow [1, k]$ be the function which assigns the clusters $\seta_1, \ldots, \seta_k$ to the ground truth clusters such that
    \[
        \sigma(i) = \argmin_{j \in [k]} \norm{\pj - \vecc_i}.
        \]
    Then, it holds by Lemma~\ref{lem:cost_lower_bound} that
    \begin{equation} \label{eq:miscl}
        \epsilon(\sigma) \triangleq \sum_{i = 1}^k \frac{\vol(\setm_{\sigma, i} \triangle \sets_i)}{\vol(\sets_i)} \leq 64 \cdot (1 + \APT)\cdot \frac{k}{\Upsilon(k)}.
    \end{equation}
    Now, assume that $\sigma$ is not a permutation, and we'll apply the following procedure inductively to construct a permutation from  $\sigma$. Since $\sigma$ isn't a permutation, there is $i\in[k]$ such that $M_{\sigma, i } = \emptyset$. Therefore, there are different values of $x,y\in[k]$ such that $\sigma(x)=\sigma(y)=j$ for some $j\neq i$.  Based on this, we construct the function $\sigma':[k]\rightarrow[k]$ such that $\sigma'(z)=i$ if $z=x$, and $\sigma'(z)=\sigma(z)$ for any other $z\in [k]\setminus \{x\}$. Notice that we can construct $\sigma'$ in this way as long as $\sigma$ isn't a permutation.
    By the definition of $\epsilon(\cdot)$ and function $\sigma'$, the difference between $\epsilon(\sigma')$ and $\epsilon(\sigma)$ can be written as
   \begin{equation} \label{eq:sigmadiff}
       \epsilon(\sigma') - \epsilon(\sigma) = \underbrace{ \left(\frac{\voldiffip}{\vol(\sets_i)} - \frac{\voldiffi}{\vol(\sets_i)}\right)}_{=: \alpha} + \underbrace{\left(\frac{\voldiffjp}{\vol(\sets_j)} - \frac{\voldiffj}{\vol(\sets_j)}\right)}_{=: \beta}.
   \end{equation} 
   Let us consider $4$ cases based on the signs of $\alpha,\beta$ defined above.
   In each case, we   bound the cost introduced by the change from $\sigma$ to $\sigma'$, and then consider the total cost introduced throughout the entire procedure of constructing a permutation.

   \textbf{Case~1:  $\alpha<0, \beta<0$.}
   In this case, it is clear that $\epsilon(\sigma') - \epsilon(\sigma) \leq 0$, and hence the total introduced  cost 
   is at most $0$.
   
   \textbf{Case~2: $\alpha>0, \beta<0$.}
   In this case, we have
   \begin{align*}
       \lefteqn{\epsilon(\sigma') - \epsilon(\sigma)}\\
       & \leq \frac{1}{\min(\vol(\sets_i), \vol(\sets_j))} \left( \voldiffip - \voldiffi + \abs{\voldiffjp - \voldiffj} \right) \\
       & = \frac{1}{\min(\vol(\sets_i), \vol(\sets_j))} \left( \voldiffip - \voldiffi + \voldiffj - \voldiffjp \right) \\
       & = \frac{1}{\min(\vol(\sets_i), \vol(\sets_j))} \left( \vol(\seta_x \setminus \sets_i) - \vol(\seta_x \intersect \sets_i) + \vol(\seta_x \setminus \sets_j) - \vol(\seta_x \intersect \sets_j) \right) \\
       & \leq \frac{2 \cdot \vol(\seta_x \setminus \sets_j)}{\min(\vol(\sets_i), \vol(\sets_j))} \\
       & \leq \frac{8 \cdot \vol(\seta_x \setminus \sets_j)}{\vol(\sets_j)},
   \end{align*}
   where the last inequality follows by the  fact that the clusters are almost balanced.
   Since each set $\seta_x$ is moved at most once in order to construct a permutation, the total introduced cost due to this case  is at most
   \begin{align*}
       \sum_{j = 1}^k \sum_{\seta_x \in \setm_{\sigma, j}} \frac{8 \cdot \vol(\seta_x \setminus \sets_j)}{\vol(\sets_j)} \leq 8 \cdot \sum_{j = 1}^k \frac{\voldiffj}{\vol(\sets_j)} \leq 512\cdot \left(1 + \APT\right)\cdot  \frac{k}{\Upsilon(k)}.
   \end{align*}
   
   \textbf{Case~3: $\alpha>0, \beta>0$.}
   In this case, we have
   \begin{align*}
       \lefteqn{\epsilon(\sigma') - \epsilon(\sigma)}\\
       & \leq \frac{1}{\min(\vol(\sets_i), \vol(\sets_j))} \left( \voldiffip - \voldiffi + \voldiffjp - \voldiffj \right) \\
        & = \frac{1}{\min(\vol(\sets_i), \vol(\sets_j))} \left(  \vol(\seta_x \setminus \sets_i) - \vol(\seta_x \cap \sets_i) + \vol(\seta_x \cap \sets_j) -  \vol(\seta_x \setminus \sets_j)   \right) \\
        & \leq  \frac{1}{\min(\vol(\sets_i), \vol(\sets_j))} \left(2\cdot  \vol(\sets_x \cap \sets_j) \right) \\
        & \leq \frac{2\cdot  \vol(\sets_j)}{\min(\vol(\sets_i), \vol(\sets_j))} \\
        & \leq 8,
   \end{align*}
   where the last inequality follows by  the fact that the clusters are almost balanced.
   We will consider the total introduced cost due to this case and Case~4 together, and   so let's first examine Case~4.
   
   \textbf{Case~4: $\alpha<0,\beta>0$.}
   In this case, we have
   \begin{align*}
       \epsilon(\sigma') - \epsilon(\sigma) & \leq \frac{1}{\vol(\sets_j)} \left(\voldiffjp - \voldiffj\right) \\
       & \leq \frac{1}{\vol(\sets_j)} \left( \vol(\seta_x \intersect \sets_j) - \vol(\seta_x \setminus \sets_j)\right)\\
       & \leq \frac{\vol(\sets_j)}{\vol(\sets_j)} \\
       & = 1.
   \end{align*}
   Now, let us bound the total number of times we need to construct $\sigma'$ in order to obtain a permutation.
   For any $i$  with $\setm_{\sigma,i}=\emptyset$,  we have
    \[
        \frac{\vol(\setm_{\sigma, i} \triangle \sets_i)}{\vol(\sets_i)} = \frac{\vol(\sets_i)}{\vol(\sets_i)} = 1,
    \]
    so the total number of required iterations is upper bounded  by
    \begin{align*}
        \cardinality{\{i : \setm_{\sigma,i}=\emptyset \}} & \leq \sum_{i = 1}^k \frac{\vol(\setm_{\sigma, i} \triangle \sets_i)}{\vol(\sets_i)} \\
        & \leq 64\cdot (1 + \APT)\cdot  \frac{k}{\Upsilon(k)}.
    \end{align*}
As such, the total introduced cost due to Cases~3 and 4 is at most  
    \[
        8 \cdot 64\cdot (1 + \APT)\cdot  \frac{k}{\Upsilon(k)} = 512\cdot (1 + \APT) \cdot \frac{k}{\Upsilon(k)}.
    \]
    Putting everything together, we have that  
    \begin{align*}
        \epsilon(\sigma^\star)&  \leq \epsilon(\sigma) + 1024\cdot (1 + \APT)\cdot \frac{k}{\Upsilon(k)}\\
        & \leq 1088\cdot (1 + \APT)\cdot  \frac{k}{\Upsilon(k)}.
    \end{align*}
    This implies that
    \[
        \sum_{i = 1}^k \vol(\setm_{\sigma^\star, i} \triangle \sets_i) \leq 2176\cdot (1 + \APT)\cdot  \frac{\vol(\vertexsetg)}{\Upsilon(k)}
    \]
    and completes the proof.
\end{proof}

We remark that this method of upper bounding  the ratio of  misclassified vertices is very different from the ones used in previous references, e.g.,~\cite{deySpectralConcentrationGreedy2019,mizutaniImprovedAnalysisSpectral2021,pengPartitioningWellClusteredGraphs2017}. In particular, instead of examining all the possible mappings between $\{\seta_i\}_{i=1}^k$ and $\{\sets_i\}_{i=1}^k$, 
we directly work with   some specifically defined function $\sigma$, and construct our desired mapping $\sigma^{\star}$ from $\sigma$. This is another key for us to obtain stronger results than the previous work.

\section{Beyond the Classical Spectral Clustering\label{sec:beyond}}
In this section we propose a variant of spectral clustering  which employs fewer than $k$ eigenvectors to find $k$ clusters. We prove that,
when the structure among the optimal clusters in an input graph satisfies certain conditions,
spectral clustering with fewer eigenvectors is able to  produce better results than classical spectral clustering.
Our result gives a theoretical justification of the surprising showcase in Section~\ref{sec:introduction},
and presents a significant speedup on the runtime of spectral clustering in practice, since fewer eigenvectors are used to construct the embedding.
 
\subsection{Encoding the Cluster-Structure into  Meta-Graphs}\label{sec:structure++}

Suppose that $\{\sets_i\}_{i=1}^k$ is a $k$-way partition of $\vertexsetg$ for an input graph $\graphg$ that minimises the $k$-way expansion $\rho(k)$. We define the  matrix $\adj_{\graphm}\in\R^{k\times k}$ by
\[
    \adj_\graphm(i, j) = \twopartdef{\weight(\sets_i, \sets_j)}{i \neq j,}{2 \weight(\sets_i, \sets_j)}{i = j}
\]
and, taking $\adj_\graphm$ to be the adjacency matrix, this 
defines a graph $\graphm = (\vertexset_\graphm, \edgeset_\graphm, \weight_\graphm)$ which we refer to as the \emph{meta-graph} of the clusters.
We define the normalised adjacency matrix of $\graphm$ by 
\[\adjn_\graphm \triangleq \degmhalfneg_\graphm \adj_\graphm \degmhalfneg,\]
and the normalised Laplacian matrix of $\graphm$ by \[\lapn_\graphm \triangleq \identity - \adjn_\graphm.\]
Let the eigenvalues of $\lapn_\graphm$ be $\gamma_1\leq\ldots\leq\gamma_k$, and $\vecg_i\in\R^k$ be the eigenvector corresponding to $\gamma_i$ for any $i\in[k]$.

The starting point of our novel approach is to look at the structure of the meta-graph $\graphm$ defined by $\{\sets_i\}_{i=1}^k$ of $\graphg$, and study  how the spectral information of $\lapn_\graphm\in\R^{k\times k}$ is encoded in the bottom eigenvectors of $\lapn_\graphg$.
To achieve this, 
 for any    $ \ell\in [k]$ and vertex $i\in \vertexset_\graphm$, let
\begin{equation}\label{eq:defxi}
\bar{\vecx}^{(i)} \triangleq \left( \vecg_1(i),\ldots, \vecg_{\ell}(i) \right)^\transpose;
\end{equation}
notice that  $\bar{\vecx}^{(i)}\in\R^{\ell}$ defines the spectral embedding of  $i\in \vertexset_\graphm$ through the bottom $\ell$ eigenvectors of $\lapn_\graphm$.

\begin{definition}[$(\theta,\ell)$-distinguishable graph]\label{def:distinguishable} For any  $\graphm=(\vertexset_\graphm, \edgeset_\graphm, \weight_\graphm)$ with  $k$ vertices, $\ell\in[k]$, and $\theta\in\R^+$, we say that $\graphm$ is  $(\theta,\ell)$-distinguishable if 
\begin{itemize}
    \item it holds for any $i\in[k]$ that  $\norm{\barx^{(i)}}^2 \geq \theta$, and 
    \item it holds for any different $i,j\in[k]$  that  \[
    \left\|\frac{\barx^{(i)}}{\norm{\barx^{(i)}}} - \frac{\barx^{(j)}}{\norm{\barx^{(j)}}}\right\|^2\geq \theta.\]
\end{itemize}
\end{definition}
In other words, graph $\graphm$ is $(\theta,\ell)$-distinguishable if (i) every embedded point $\bar{\vecx}^{(i)}$ has squared length at least $\theta$, and (ii) any pair of embedded points with normalisation are separated by
 a distance
of at least $\theta$.
By definition, it is easy to see that, if $\graphm$ is $(\theta,\ell)$-distinguishable for some large value of $\theta$, then the embedded points $\{\bar{\vecx}^{(i)} \}_{i\in \vertexset_\graphm}$ can be easily separated even if $\ell< k$.
The two examples below demonstrate that it is indeed the case and, since the meta-graph $\graphm$ is constructed from  $\{\sets_i\}_{i=1}^k$, this well-separation property for $\{\bar{\vecx}^{(i)} \}_{i\in \vertexset_\graphm}$ usually implies 
that the clusters $\{\sets_i\}_{i = 1}^k$ are also well-separated when the vertices are mapped to the points $\{F(u)\}_{u \in \vertexset_\graphg}$, in which
\begin{equation}\label{eq:embeddingell}
   F(u)\triangleq \frac{1}{\sqrt{\deg(u)}}\cdot  \left(\vecf_1(u), \ldots \vecf_{\ell}(u)\right)^\transpose.
\end{equation}

\begin{example} \label{ex:cycle}
 Suppose the meta-graph is $C_6$, the cycle on $6$ vertices.
 Figure~\ref{fig:cycle_example}(\subref{subfig:cycle_embedding}) shows that the vertices of $C_6$ are well separated by the second and third eigenvectors of $\lapn_{C_6}$.\footnote{Notice that the first eigenvector is the trivial one and gives no useful information. This is why we visualise the second and third eigenvectors only.}
 Since the minimum distance between any pair of vertices in this embedding is $2/3$, we say that $C_6$ is $\left(2/3, 3\right)$-distinguishable.
 Figure~\ref{fig:cycle_example}(\subref{subfig:cycle_sbm_embedding}) shows that, when using $\vecf_2, \vecf_3$ of $\lapn_G$ to embed the vertices of a $600$-vertex graph with a cyclical cluster pattern, the embedded points closely match the ones from the meta-graph.
\end{example}
 
\begin{figure}[ht]
\centering
\begin{subfigure}{0.35\textwidth}
        \includegraphics[width=\textwidth]{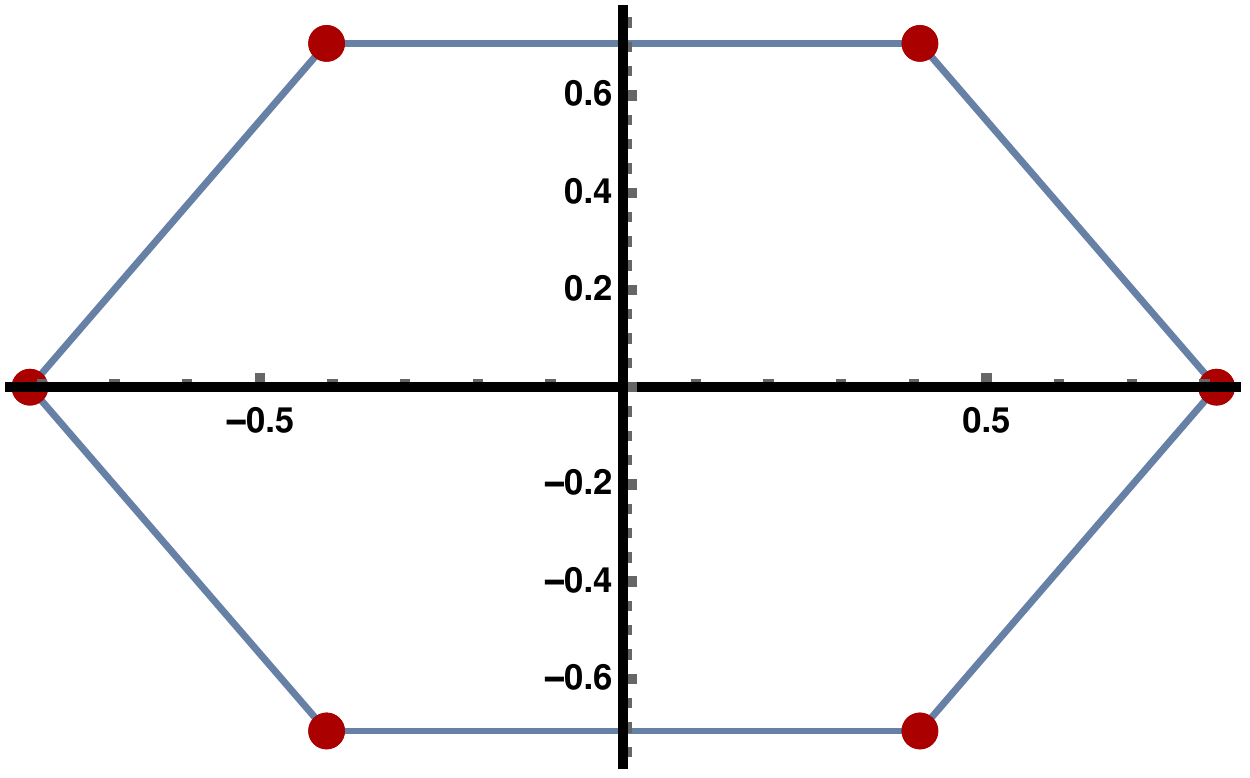}%
    \caption{
    Spectral embedding of $C_6$ with $\vecg_2$ and $\vecg_3$.
    \label{subfig:cycle_embedding}   
    }
\end{subfigure}
\hspace{5em}
\begin{subfigure}{0.35\textwidth}
    \includegraphics[width=\textwidth]{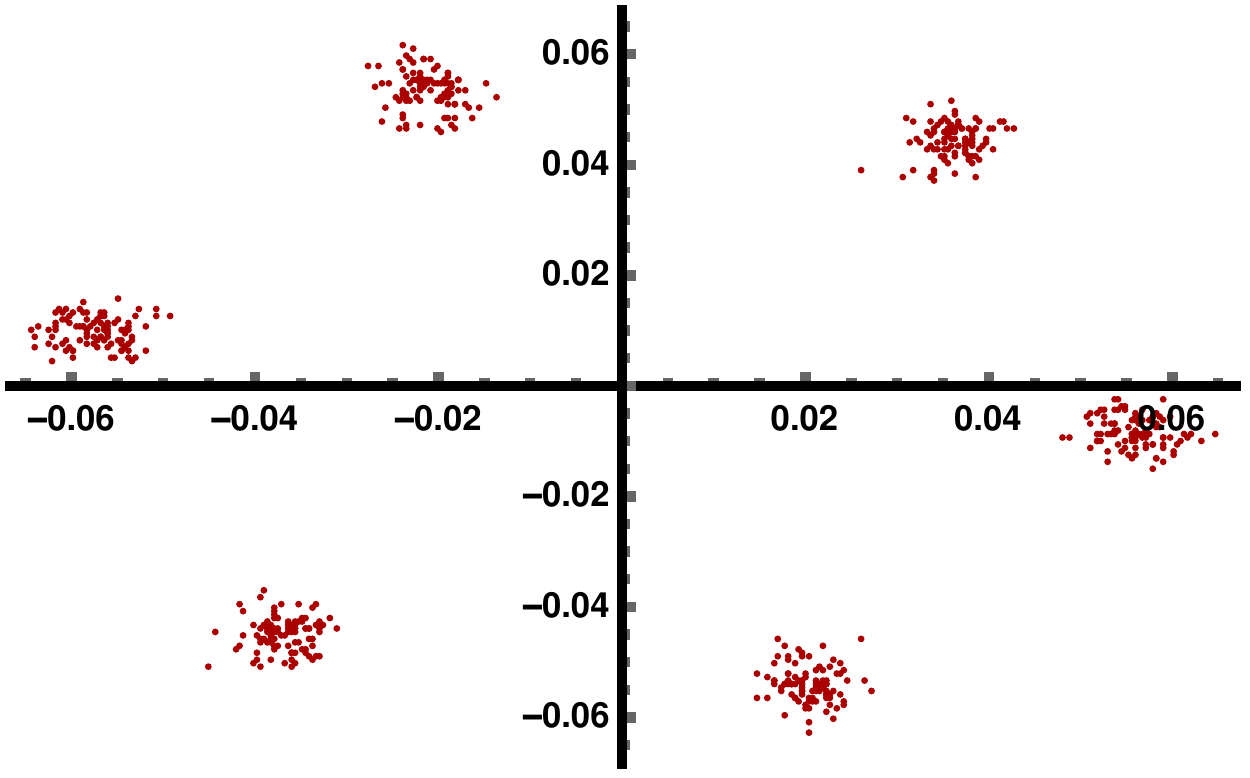}%
    \caption{Spectral embedding of $\graphg$ with $\vecf_2$ and $\vecf_3$.
    \label{subfig:cycle_sbm_embedding}
    }
\end{subfigure}
    \caption[Spectral embedding of a graph whose clusters exhibit a cyclical structure]{The spectral embedding of a large graph $\graphg$ whose clusters exhibit a cyclical structure closely matches the embedding of the meta-graph $C_6$.
    \label{fig:cycle_example}
    }
\end{figure}
 
\begin{example} \label{ex:grid}
Suppose the meta-graph is $P_{4, 4}$, which is the $4 \times 4$ grid graph.
Figure~\ref{fig:grid_example}(\subref{subfig:grid_embedding}) shows that the vertices are separated using the second and third eigenvectors of $\lapn_{P_{4, 4}}$.
The minimum distance between any pair of vertices in this embedding is roughly $0.1$, and so $P_{4, 4}$ is $(0.1, 3)$-distinguishable.
 Figure~\ref{fig:grid_example}(\subref{subfig:grid_sbm_embedding}) demonstrates that,
 when using $\vecf_2, \vecf_3$ of $\lapn_\graphg$ to construct the embedding, the embedded points    closely match the ones from the meta-graph.
\end{example}

\begin{figure}[ht]
\centering
\begin{subfigure}{0.35\textwidth}
        \includegraphics[width=\textwidth]{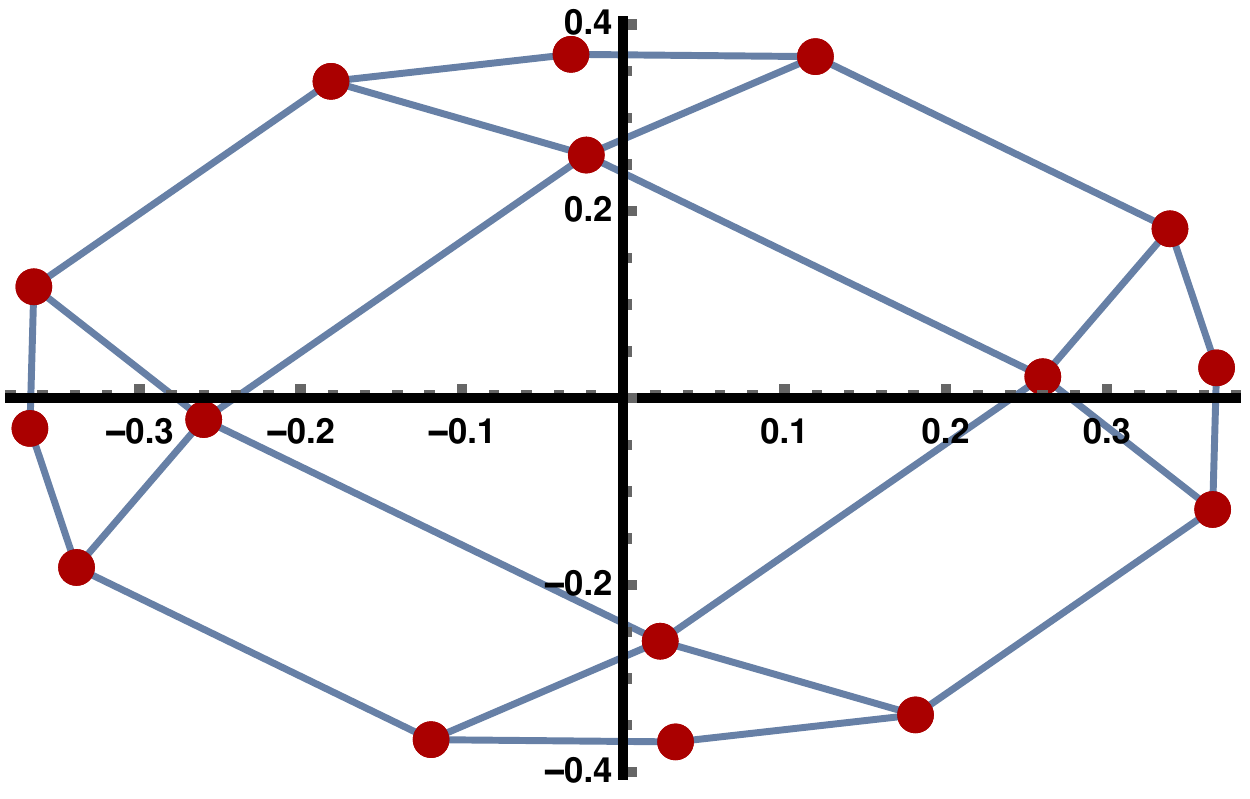}%
        \caption{
        Embedding of $P_{4, 4}$ with $\vecg_2$ and $\vecg_3$.
        \label{subfig:grid_embedding}
        }
\end{subfigure}
\hspace{5em}
\begin{subfigure}{0.35\textwidth}
    \includegraphics[width=\textwidth]{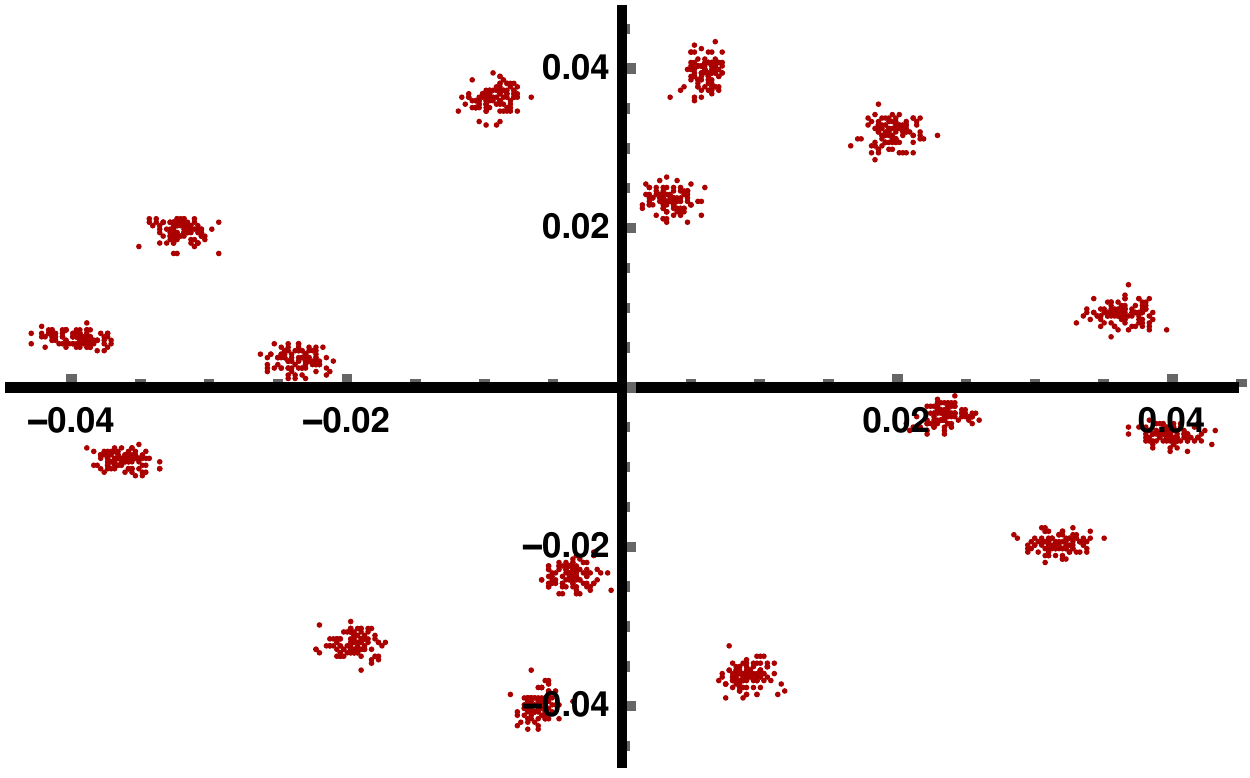}
    \caption{Embedding of $\graphg$ with $\vecf_2$ and $\vecf_3$.
    \label{subfig:grid_sbm_embedding}    
    }
\end{subfigure}
    \caption[Spectral embedding of a graph whose clusters exhibit a grid structure]{
    The spectral embedding of a large $1600$-vertex graph $\graphg$ whose clusters exhibit a grid structure closely matches the embedding of the meta-graph $P_{4, 4}$.
    \label{fig:grid_example}
    }
\end{figure}
 
From these examples, it is clear to see that   there is a close
connection between the embedding $\{\bar{\vecx}^{(i)}\}$ defined in \eqref{eq:defxi} and the embedding $\{F(u)\}$ defined in \eqref{eq:embeddingell}. To formally analyse this connection, we study the structure theorem with meta-graphs. 
 
We   define vectors $\barg_1, \ldots, \barg_k \in \R^n$ which represent the eigenvectors of $\lapn_\graphm$ ``blown-up'' to $n$ dimensions.
Formally, we define $\barg_i$ such that
\[
    \barg_i = \sum_{j = 1}^{k} \frac{\degmhalf \indicatorvec_j}{\norm{\degmhalf \indicatorvec_j}}\cdot \vecg_i(j),
\]
where $\indicatorvec_j \in \R^n$ is the indicator vector of cluster $\sets_j$ defined in \eqref{eq:const_g}.
By definition, it holds for any $u\in \sets_j$ that 
\[
    \barg_i(u) = \sqrt{\frac{\deg(u)}{\vol(\sets_j)}}\cdot \vecg_i(j).
\]
The following lemma shows  that $\{\bar{\vecg}_i\}_{i=1}^k$ form an orthonormal basis.

\begin{lemma} \label{lem:barg_ortho}
The following statements hold:
\begin{enumerate}
    \item it holds for any $i\in[k]$ that  $\norm{\barg_i} = 1$;
    \item it holds for any different  $i, j\in [k]$ that  $\langle \barg_i, \barg_j\rangle = 0$.
\end{enumerate}
\end{lemma}
\begin{proof}
By definition, we have that 
\begin{align*}
        \norm{\barg_i}^2  = \barg_i^\transpose \barg_i & = \sum_{j = 1}^k \sum_{u \in \sets_j} \barg_i(u)^2  = \sum_{j=1}^k \sum_{u \in \sets_j} \left(\frac{\sqrt{\deg(u)}}{\sqrt{\vol(\sets_j)}} \cdot  \vecg_i(j) \right)^2  = \sum_{j = 1}^k \vecg_i(j)^2 = \norm{\vecg_i}^2 = 1,
    \end{align*}
    which proves the first statement. 

    To prove the second statement, we have for any $i \neq j$ that 
    \begin{align*}
        \langle \barg_i, \barg_j\rangle & = \sum_{x = 1}^{k} \sum_{u \in \sets_x} \barg_i(u) \barg_j(u) \\
        & = \sum_{x = 1}^{k} \sum_{u \in \sets_x} \frac{\deg(u)}{\vol(\sets_x)}\cdot \vecg_i(x) \vecg_j(x) \\
        & = \sum_{x = 1}^k \vecg_i(x) \vecg_j(x)\\
        & = \vecg_i^\transpose \vecg_j \\
        & = 0,
    \end{align*}
    which completes the proof.
\end{proof}

We will later use the following important relationship between the eigenvalues of $\lapn_\graphm$ and $\lapn_\graphg$.

\begin{lemma} \label{lem:lambda_leq_gamma2}
 It holds for any $1\leq i\leq k$ that  $\lambda_i \leq \gamma_i$.
\end{lemma}

\begin{proof}
    Notice that we have for any $j \leq k$ that 
    \begin{align*}
        \barg_j^\transpose \lapn_\graphg \barg_j & = \sum_{(u, v) \in \edgesetg} \weight_\graphg(u, v) \left(\frac{\barg_j(u)}{\sqrt{\deg(u)}} - \frac{\barg_j(u)}{\sqrt{\deg(v)}}\right)^2 \\
        & = \sum_{x = 1}^{k - 1} \sum_{y = x + 1}^k \sum_{a \in \sets_x} \sum_{b \in \sets_y} \weight(a, b) \left(\frac{\barg_j(a)}{\sqrt{\deg(a)}} - \frac{\barg_j(b)}{\sqrt{\deg(b)}} \right)^2 \\
        & = \sum_{x = 1}^{k - 1} \sum_{y = x + 1}^k \weight(\sets_x, \sets_y) \left(\frac{\vecg_j(x)}{\sqrt{\vol(\sets_x)}} - \frac{\vecg_j(y)}{\sqrt{\vol(\sets_y)}} \right)^2 \\
        &  = \vecg_j^\transpose \lapn_\graphm \vecg_j \\
        & = \gamma_j. 
    \end{align*}
    By Lemma~\ref{lem:barg_ortho}, we have an $i$-dimensional subspace $\sets_i \subset \R^n$ such that
    \[
        \max_{\vecx \in \sets_i} \frac{\vecx^\transpose \lapn_\graphg \vecx}{\vecx^\transpose \vecx} = \gamma_i,
    \]
    from which the statement of the lemma follows by the Courant-Fischer theorem.
\end{proof}

Next, similar to the function $\Upsilon(k)$
defined in \eqref{eq:defineupsilon}, for any input graph $\geqvewg$ and   $(\theta, \ell)$-distinguishable meta-graph $\graphm$, we define the function $\Psi(\ell)$ 
by
\[
        \Psi(\ell) \triangleq \sum_{i = 1}^{\ell} \frac{\gamma_i}{\lambda_{\ell + 1}}.
    \]
Notice that we have by the higher-order Cheeger inequality that $\gamma_i/2 \leq  \rho_\graphm(i)$ holds for any $i\in[\ell]$, and $\rho_\graphm(i)\leq \rho_\graphg(k)$ by the construction of matrix $\adj_\graphm$. Hence, one can view $\Psi(\ell)$ as a refined definition  of $\Upsilon(k)$.

We now show that the vectors $\vecf_1, \ldots, \vecf_{\ell}$ and $\barg_1, \ldots, \barg_{\ell}$ are well approximated by each other.
In order to show this, we define for any $ i\in [\ell]$ the vectors
\[
    \hatf_i = \sum_{j = 1}^{\ell} \langle \barg_i, \vecf_j \rangle \vecf_j~\mbox{\ \ \ and\ \ \ }~\hatg_i = \sum_{j = 1}^{\ell} \langle \vecf_i, \barg_j \rangle \barg_j,
\]
and present the structure theorem with meta-graphs.

\begin{theorem}[The Structure Theorem with Meta-Graphs] \label{thm:structure++} 
The following statements hold:
\begin{enumerate}
    \item it holds for  any $i \in [\ell]$ that \[\norm{\barg_i - \hatf_i}^2 \leq \gamma_i / \lambda_{\ell + 1};\]
    \item it holds for any $ \ell\in [k]$ that 
    \[
        \sum_{i = 1}^{\ell} \norm{\vecf_i - \hatg_i}^2 \leq \sum_{i = 1}^{\ell} \frac{\gamma_i}{\lambda_{\ell + 1}}.
    \]
\end{enumerate}

\end{theorem}
\begin{proof}
For the first statement, we    write $\barg_i$ as a linear combination of the vectors $\vecf_1, \ldots, \vecf_n$, i.e., 
    $
        \barg_i = \sum_{j = 1}^n \langle \barg_i, \vecf_j \rangle \vecf_j.$
    Since  $\hatf_i$ is a projection of $\barg_i$, we have that $\barg_i - \hatf_i$ is perpendicular to $\hatf_i$, and that
    \begin{align*}
        \norm{\barg_i - \hatf_i}^2 & = \norm{\barg_i}^2 - \norm{\hatf_i}^2 = \left(\sum_{j = 1}^n \langle \barg_i, \vecf_j \rangle^2 \right) - \left(\sum_{j = 1}^{\ell} \langle \barg_i, \vecf_j \rangle^2 \right) = \sum_{j = \ell + 1}^n \langle \barg_i, \vecf_j \rangle^2.
    \end{align*}
    Now, we study the quadratic form $\barg_i^\transpose \lapn_\graphg \barg_i$ and have that 
    \begin{align*}
        \barg_i^\transpose \lapn_\graphg \barg_i & = \left(\sum_{j = 1}^n \langle \barg_i, \vecf_j \rangle \vecf_j^\transpose \right) \lapn_\graphg \left(\sum_{j = 1}^n \langle \barg_i, \vecf_j \rangle \vecf_j\right) = \sum_{j = 1}^n \langle \barg_i, \vecf_j \rangle^2 \lambda_j \geq \lambda_{\ell + 1} \norm{\barg_i - \hatf_i}^2.
    \end{align*}
    By the proof of Lemma~\ref{lem:lambda_leq_gamma2}, we have that $\barg_i^\transpose \lapn_\graphg \barg_i \leq \gamma_i$, from which the first statement follows.
   
    Now we prove the second statement. 
    We define the vectors  $\barg_{k+1},\ldots, \barg_n$   to be an arbitrary orthonormal basis of the space orthogonal to the space spanned by 
    $\barg_1, \ldots, \barg_k$. Then, we can write any $\vecf_i$ as
    $
        \vecf_i = \sum_{j = 1}^n \langle \vecf_i, \barg_j \rangle \barg_j$,
    and   have that 
    \begin{align*}
        \sum_{i = 1}^{\ell} \norm{\vecf_i - \hatg_i}^2 & = \sum_{i = 1}^{\ell} \left(\norm{\vecf_i}^2 - \norm{\hatg_i}^2 \right) \\
        & = \ell - \sum_{i = 1}^{\ell} \sum_{j = 1}^{\ell} \langle \vecf_i, \barg_j \rangle^2 \\
        & = \sum_{j = 1}^{\ell} \left( 1 - \sum_{i = 1}^{\ell} \langle \barg_j, \vecf_i \rangle^2 \right) \\
        & = \sum_{j = 1}^{\ell} \left( \norm{\barg_j}^2 - \norm{\hatf_j}^2 \right) \\
        & = \sum_{j = 1}^{\ell} \norm{\barg_j - \hatf_j}^2 \\
        & \leq \sum_{j = 1}^{\ell} \frac{\gamma_j}{\lambda_{{\ell} + 1}},
    \end{align*}
    where the final inequality follows by
    the first statement of the theorem.
\end{proof}

\subsection{Spectral Clustering with Fewer Eigenvectors \label{sec:embedding+}}
In this section, we analyse spectral clustering with fewer eigenvectors. Our presented algorithm is essentially the same as the standard spectral clustering   described in Section~\ref{sec:analysis1}, with the only difference that every   $u\in \vertexsetg$ is embedded into a point in $\R^{\ell}$ by the mapping
defined in
\eqref{eq:embeddingell}. Our analysis follows from the one from Section~\ref{sec:analysis1} at a very high level.
However, since we require that $\{F(u)\}_{u \in \vertexsetg}$ are well separated in $\R^{\ell}$ for some $\ell < k$, the proof is more involved. 

For any  $i \in [k]$, we  define the approximate centre $\vecp^{(i)} \in \R^{\ell}$ of every cluster $\sets_i$ by
\[
    \vecp^{(i)}(j) = \frac{1}{\sqrt{\vol(\sets_i)}}\cdot \sum_{x = 1}^{\ell} \langle \vecf_j, \barg_x \rangle\cdot  \vecg_x(i),
\]
and prove that the total $k$-means cost for the points $\{F(u)\}_{u\in \vertexsetg}$ can be upper bounded.

\begin{lemma} \label{lem:cost_bound}
It holds that 
    \[
        \sum_{i = 1}^k \sum_{u \in \sets_i} \deg(u) \norm{F(u) - \peye}^2 \leq \Psi(\ell).
    \]
\end{lemma}
\begin{proof}
By definition, it holds that 
    \begin{align*}
        \sum_{i = 1}^k \sum_{u \in S_i} \deg(u) \norm{F(u) - \peye}^2 
      & = \sum_{i = 1}^k \sum_{u \in \sets_i} \deg(u) \left[ \sum_{j = 1}^{\ell} \left( \frac{\vecf_j(u)}{\sqrt{\deg(u)}} - \left(\sum_{x = 1}^{\ell} \langle \vecf_j, \barg_x \rangle \frac{\vecg_x(i)}{\sqrt{\vol(\sets_i)}}  \right) \right)^2 \right] \\
        & = \sum_{i = 1}^k \sum_{u \in \sets_i} \sum_{j = 1}^{\ell} \left( \vecf_j(u) - \left(\sum_{x = 1}^{\ell} \langle \vecf_j, \barg_x \rangle \barg_x(u) \right) \right)^2 \\
        & = \sum_{i = 1}^k \sum_{u \in \sets_i} \sum_{j = 1}^{\ell} \left( \vecf_j(u) - \hatg_j(u) \right)^2 \\
        & = \sum_{j = 1}^{\ell} \norm{\vecf_j - \hatg_j}^2 \leq \Psi(\ell),
    \end{align*}
    where the final inequality follows from the second statement of Theorem~\ref{thm:structure++}.
\end{proof}

We now prove a sequence of lemmas which will establish that the distance between different $\peye$ and $\pj$ can be lower bounded with respect to $\theta$ and $\Psi(\ell)$.

\begin{lemma}\label{lem:p_norm}
It holds for  $i \in [k]$ that 
    \[
       \left(1 -  \frac{4 \sqrt{\Psi(\ell)}}{\theta}\right) \frac{\norm{\barx^{(i)}}^2}{\vol(\sets_i)}  \leq \norm{\vecp^{(i)}}^2 \leq  \frac{\norm{\barx^{(i)}}^2}{\vol(\sets_i)} \left(1 +  \frac{2 \sqrt{\Psi(\ell)}}{\theta}\right).
    \]
\end{lemma}
\begin{proof}
It holds by definition that 
    \begin{align}
        \vol(\sets_i)\cdot  \norm{\vecp^{(i)}}^2 
        & = \sum_{x = 1}^{\ell} \left( \sum_{y = 1}^{\ell } \langle \vecf_x, \barg_y \rangle \vecg_y(i)\right)^2  \nonumber \\
        & = \sum_{x = 1}^{\ell} \sum_{y = 1}^{\ell} \sum_{z = 1}^{\ell} \langle \vecf_x, \barg_y\rangle \langle \vecf_x, \barg_z \rangle   \vecg_y(i) \vecg_z(i) \nonumber \\
        & = \sum_{x = 1}^{\ell} \sum_{y = 1}^{\ell} \langle \vecf_x, \barg_y\rangle^2 \vecg_y(i)^2 + \sum_{x = 1}^{\ell} \sum_{y = 1}^{\ell} \sum_{\substack{z = 1\\z \neq y}}^{\ell} \langle \vecf_x, \barg_y\rangle \langle \vecf_x, \barg_z\rangle \vecg_y(i) \vecg_z(i) \nonumber \\
         & = \sum_{x = 1}^{\ell} \sum_{y = 1}^{\ell} \langle \vecf_x, \barg_y\rangle^2 \vecg_y(i)^2 + \sum_{x = 1}^{\ell} \sum_{y = 1}^{\ell} \sum_{\substack{z = 1\\z \neq y}}^{\ell}  \vecg_y(i) \vecg_z(i) \cdot \left( \hatf_y^\transpose \hatf_z\right).  \label{eq:separate}
    \end{align}
    We study the two terms of \eqref{eq:separate} separately. For the second term, we have that 
    \begin{align*}
        \lefteqn{\sum_{x = 1}^{\ell} \sum_{y = 1}^{\ell} \sum_{\substack{z = 1\\z \neq y}}^{\ell}  \vecg_y(i) \vecg_z(i) \cdot \left( \hatf_y^\transpose \hatf_z\right)}\\
        & \leq \sum_{y = 1}^{\ell} \abs{\vecg_y(i)} \sum_{\substack{1\leq z \leq\ell \\ z \neq y}} \abs{\vecg_z(i)} \abs{\hatf_y^\transpose \hatf_z} \\
        & = \sum_{y = 1}^{\ell} \abs{\vecg_y(i)} \sum_{\substack{1\leq z\leq \ell \\ z \neq y}} \abs{\vecg_z(i)} \abs{ \left(\barg_y + \hatf_y - \barg_y\right)^\transpose \left(\barg_z + \hatf_z - \barg_z\right) } \\
        & = \sum_{y = 1}^{\ell} \abs{\vecg_y(i)} \sum_{\substack{1\leq z\leq \ell \\ z \neq y}} \abs{\vecg_z(i)} \abs{ \langle \hatf_y - \barg_y, \barg_z \rangle + \langle \hatf_z - \barg_z, \barg_y \rangle + \langle \hatf_y - \barg_y, \hatf_z - \barg_z \rangle } \\
        & =  \sum_{y = 1}^{\ell} \abs{\vecg_y(i)} \sum_{\substack{z = 1 \\ z \neq y}}^{\ell} \abs{\vecg_z(i)} \abs{ \langle \hatf_y - \barg_y, \barg_z \rangle } \\
        & \leq  \sqrt{\left(\sum_{y = 1}^{\ell} \vecg_y(i)^2\right) \sum_{y = 1}^{\ell} \left(\sum_{\substack{1\leq z\leq \ell \\ z \neq y}} \abs{\vecg_z(i)} \abs{ \langle \hatf_y - \barg_y, \barg_z \rangle }\right)^2 } \\
        & \leq   \sqrt{\sum_{y = 1}^{\ell} \left(\sum_{\substack{1\leq z\leq \ell \\ z \neq y}} \vecg_z(i)^2 \right) \left( \sum_{\substack{1\leq  z\leq \ell  \\ z \neq y}} \langle \hatf_y - \barg_y, \barg_z \rangle^2 \right) } \\
        & \leq   \sqrt{\sum_{y = 1}^{\ell} \sum_{\substack{1\leq z\leq \ell \\ z \neq y}} \langle \hatf_y - \barg_y, \barg_z \rangle^2 } \\
        & \leq  \sqrt{\sum_{y = 1}^{\ell} \norm{\hatf_y - \barg_y}^2 }  \leq   \sqrt{\Psi(\ell)},
    \end{align*}
    where we used the fact that $\sum_{y = 1}^k \vecg_y(i)^2 = 1$ for all $i\in[k]$. Therefore, we have that
    \begin{align*}
        \vol(\sets_i) \norm{\vecp^{(i)}}^2 & \leq \sum_{y = 1}^{\ell} \left(\sum_{x = 1}^{\ell} \langle \vecf_x, \barg_y\rangle^2 \right) \vecg_y(i)^2 +   \sqrt{\Psi (\ell)} \leq \sum_{y = 1}^{\ell} \vecg_y(i)^2 +   \sqrt{\Psi (\ell)}  \leq \norm{\barx^{(i)}}^2 +   \sqrt{\Psi(\ell)} \\
        & \leq \norm{\barx^{(i)}}^2 \left(1 + \frac{2 \sqrt{\Psi(\ell)}}{\theta}\right).
    \end{align*}
    On the other hand, we have that 
    \begin{align*}
        \vol(\sets_i) \norm{\vecp^{(i)}}^2 & \geq \sum_{y = 1}^{\ell} \left(\sum_{x = 1}^{\ell} \langle \vecf_x, \barg_y\rangle^2 \right) \vecg_y(i)^2 - 2 \sqrt{\Psi(\ell)}  \geq \sum_{y = 1}^{\ell} \norm{\hatf_y}^2 \vecg_y(i)^2 - 2 \sqrt{\Psi(\ell)} \\
        & \geq \left(1 - \Psi(\ell)\right) \norm{\barx^{(i)}}^2 - 2 \sqrt{\Psi(\ell)} \geq \norm{\barx^{(i)}}^2 \left(1 - \frac{4 \sqrt{\Psi(\ell)}}{\theta}\right), 
    \end{align*}
    where the last inequality holds by the fact that  $\left\|\barx^{(i)} \right\|\leq 1$ and $\Psi(\ell)<1$.   Hence, the statement holds.
\end{proof}

\begin{lemma} \label{lem:pnorm_diff}
    It holds for  $i \neq j$ that 
    \[
        \norm{\frac{\sqrt{\vol(\sets_i)} }{\norm{\barx^{(i)}}}\cdot \vecp^{(i)} - \frac{\sqrt{\vol(\sets_j)} }{\norm{\barx^{(j)}}} \cdot \vecp^{(j)} }^2 \geq \theta - 3 \sqrt{\Psi(\ell)}.
    \]
\end{lemma}

\begin{proof}
By definition, it holds that 
    \begin{alignat*}{2}
    \lefteqn{\norm{\frac{\sqrt{\vol(\sets_i)} }{\norm{\barxi}}\cdot \peye - \frac{\sqrt{\vol(\sets_j)} }{\norm{\barxj}}\cdot \pj }^2 }\\
    & =  \sum_{t = 1}^{\ell} \left(\sum_{y = 1}^{\ell} \langle \vecf_t, \barg_y \rangle \left(\frac{\vecg_y(i)}{\norm{\barxi}} - \frac{\vecg_y(j)}{\norm{\barxj}}\right)\right)^2 \\
    & =  \sum_{t = 1}^{\ell} \sum_{y = 1}^{\ell} \langle \vecf_t, \barg_y\rangle^2 \left(\frac{\vecg_y(i)}{\norm{\barxi}} - \frac{\vecg_y(j)}{\norm{\barxj}} \right)^2\\
    & \qquad\qquad  +  \sum_{t = 1}^{\ell} \sum_{y = 1}^{\ell} \sum_{\substack{1\leq z \leq \ell \\ z \neq y}} \langle \vecf_t, \barg_y\rangle \langle \vecf_t, \barg_z \rangle \left(\frac{\vecg_y(i)}{\norm{\barxi}} - \frac{\vecg_y(j)}{\norm{\barxj}} \right) \left(\frac{\vecg_z(i)}{\norm{\barxi}} - \frac{\vecg_z(j)}{\norm{\barxj}} \right).
    \end{alignat*}
    We upper bound  the second term by 
    \newcommand{\gnormij}[1]{\frac{\vecg_{#1}(i)}{\norm{\barx^{(i)}}} - \frac{\vecg_{#1}(j)}{\norm{\barx^{(j)}}}}
    \begin{align*}
        \lefteqn{\sum_{y = 1}^{\ell} \abs{\frac{\vecg_y(i)}{\norm{\barx^{(i)}}} - \frac{\vecg_y(j)}{\norm{\barx^{(j)}}}} \sum_{\substack{1\leq z \leq \ell \\ z \neq y}} \abs{\frac{\vecg_z(i)}{\norm{\barx^{(i)}}} - \frac{\vecg_z(j)}{\norm{\barx^{(j)}}}} \sum_{t = 1}^{\ell } \abs{\langle \vecf_t, \barg_y \rangle \langle \vecf_t, \barg_x \rangle}} \\
        = & \sum_{y = 1}^{\ell} \abs{\frac{\vecg_y(i)}{\norm{\barx^{(i)}}} - \frac{\vecg_y(j)}{\norm{\barx^{(j)}}}} \sum_{\substack{1\leq z\leq \ell  \\ z \neq y}} \abs{\frac{\vecg_z(i)}{\norm{\barx^{(i)}}} - \frac{\vecg_z(j)}{\norm{\barx^{(j)}}}} \abs{\hatf_y^\transpose \hatf_z} \\
        = & \sum_{y = 1}^{\ell} \abs{\frac{\vecg_y(i)}{\norm{\barx^{(i)}}} - \frac{\vecg_y(j)}{\norm{\barx^{(j)}}}} \sum_{\substack{1\leq z \leq \ell \\ z \neq y}} \abs{\frac{\vecg_z(i)}{\norm{\barx^{(i)}}} - \frac{\vecg_z(j)}{\norm{\barx^{(j)}}}} \abs{ \left(\barg_y + \hatf_y - \barg_y\right)^\transpose \left(\barg_z + \hatf_z - \barg_z\right) } \\
       = &   \sum_{y = 1}^{\ell} \abs{\frac{\vecg_y(i)}{\norm{\barx^{(i)}}} - \frac{\vecg_y(j)}{\norm{\barx^{(j)}}}} \sum_{\substack{1\leq z \leq \ell \\ z \neq y}} \abs{\frac{\vecg_z(i)}{\norm{\barx^{(i)}}} - \frac{\vecg_z(j)}{\norm{\barx^{(j)}}}} \abs{ \langle \hatf_y - \barg_y, \barg_z \rangle } \\
         \leq &   \sqrt{\left(\sum_{y = 1}^{\ell} \left(\gnormij{y}\right)^2\right) \sum_{y = 1}^{\ell} \left(\sum_{\substack{1\leq z \leq \ell \\ z \neq y}} \abs{\gnormij{z}} \abs{ \langle \hatf_y - \barg_y, \barg_z \rangle }\right)^2 } \\
         \leq &  \sqrt{2 \sum_{y = 1}^{\ell} \left(\sum_{\substack{1\leq z \leq \ell \\ z \neq y}} \left(\gnormij{z}\right)^2 \right) \left( \sum_{\substack{1\leq z \leq \ell \\ z \neq y}} \langle \hatf_y - \barg_y, \barg_z \rangle^2 \right) } \\
         \leq &  2 \sqrt{\sum_{y = 1}^{\ell} \sum_{\substack{1\leq z \leq \ell \\ z \neq y}} \langle \hatf_y - \barg_y, \barg_z \rangle^2 } \\
          \leq &  2 \sqrt{\sum_{y = 1}^{\ell} \norm{\hatf_y - \barg_y}^2 } 
         \leq   2 \sqrt{\Psi (\ell)},
    \end{align*}
    from which we can conclude that
    \begin{alignat*}{2}
    \norm{\frac{\sqrt{\vol(\sets_i)} }{\norm{\barxi}}\cdot \peye - \frac{\sqrt{\vol(\sets_j)} }{\norm{\barxj}}\cdot \pj}^2 & \geq && \left(1 - \Psi(\ell) \right) \theta - 2 \sqrt{\Psi(\ell)} \\
    & \geq && \theta - 3 \sqrt{\Psi(\ell)}. 
    \end{alignat*} 
    With this we proved the statement.
\end{proof}

\begin{lemma}\label{lem:norm_phat_diff}
It holds for any different  $i,j\in[k]$ that  
    \[
        \norm{\frac{\peye}{\norm{\peye}} - \frac{\pj}{\norm{\pj}}}^2 \geq \frac{\theta}{4} - 8 \sqrt{\frac{\Psi}{\theta}}.
    \]
\end{lemma}
\begin{proof}
To follow the proof, it may help to refer to the illustration in Figure~\ref{fig:proof_fig}.
    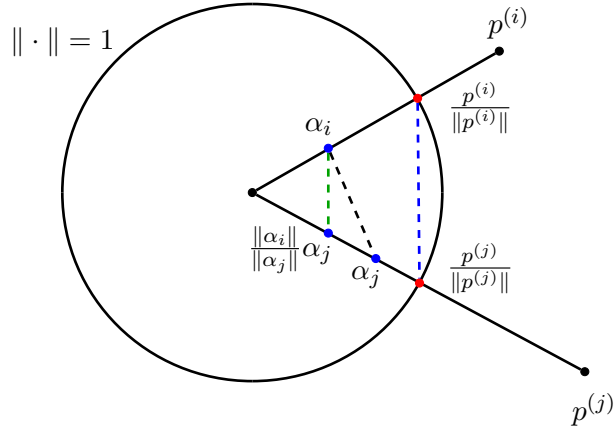
\begin{figure}[ht]
        \centering
        \begin{tikzpicture}
	\begin{pgfonlayer}{nodelayer}
		\node [style=none] (0) at (0, 5) {};
		\node [style=none] (1) at (5, 0) {};
		\node [style=none] (2) at (0, -5) {};
		\node [style=none] (3) at (-5, 0) {};
		\node [style=smallblack] (4) at (0, 0) {};
		\node [style=none] (5) at (-5, 4) {$\|\cdot\| = 1$};
		\node [style=smallblack] (7) at (8.75, -4.75) {};
		\node [style=smallblack] (8) at (6.5, 3.75) {};
		\node [style=none] (9) at (6.75, 4.5) {$p^{(i)}$};
		\node [style=none] (10) at (9, -5.75) {$p^{(j)}$};
		\node [style=smallred] (11) at (4.35, 2.5) {};
		\node [style=smallred] (12) at (4.4, -2.4) {};
		\node [style=none] (13) at (6, 2.25) {$\frac{p^{(i)}}{\|p^{(i)}\|}$};
		\node [style=none] (15) at (6, -2) {$\frac{p^{(j)}}{\|p^{(j)}\|}$};
		\node [style=smallblue] (16) at (2, 1.175) {};
		\node [style=smallblue] (17) at (3.25, -1.75) {};
		\node [style=none] (18) at (1.75, 1.75) {$\alpha_i$};
		\node [style=none] (19) at (3, -2.25) {$\alpha_j$};
		\node [style=smallblue] (20) at (2, -1.075) {};
		\node [style=none] (21) at (1, -1.5) {$\frac{\|\alpha_i\|}{\|\alpha_j\|}\alpha_j$};
	\end{pgfonlayer}
	\begin{pgfonlayer}{edgelayer}
		\draw [style=undirected, bend left=45] (0.center) to (1.center);
		\draw [style=undirected, bend left=45] (1.center) to (2.center);
		\draw [style=undirected, bend left=45] (2.center) to (3.center);
		\draw [style=undirected, bend left=45] (3.center) to (0.center);
		\draw [style=undirected] (4) to (8);
		\draw [style=undirected] (4) to (7);
		\draw [style=blue dashed] (11) to (12);
		\draw [style=green dashed] (16) to (20);
		\draw [style=undirected dashed] (16) to (17);
	\end{pgfonlayer}
\end{tikzpicture}
        \caption[Illustration of the proof of Lemma~\ref{lem:norm_phat_diff}]{
        Illustration of the proof of Lemma~\ref{lem:norm_phat_diff}.
        Our goal is to give a lower bound on the length of $\left(\frac{\peye}{\norm{\peye}} - \frac{\pj}{\norm{\pj}}\right)$, which is the blue dotted line in the figure.
        We instead calculate a lower bound on the length of $\left( \veca_i - \frac{\norm{\veca_i}}{\norm{\veca_j}} \right)$, which is the green dotted line, and use the fact that by construction, $\norm{\veca_i} \leq 1$ and $\norm{\veca_j} \leq 1$.}
        \label{fig:proof_fig}
    \end{figure}
   
    We set the parameter $\epsilon = 4 \sqrt{\Psi} / \theta$, and define 
    \[
        \veca_i = \frac{\sqrt{\vol(\sets_i)} }{\left(1 + \epsilon\right) \norm{\barxi}}\cdot \peye, \qquad \hspace{3em} \veca_j = \frac{\sqrt{\vol(\sets_j)} }{\left(1 + \epsilon\right) \norm{\barxj}}\cdot \pj.
    \]
    By the definition of $\epsilon$ and  Lemma~\ref{lem:p_norm}, it holds that 
    $
        \norm{\veca_i} \leq 1$, and $
        \norm{\veca_j} \leq  1$.
    We can also assume without loss of generality that $\norm{\veca_i} \leq \norm{\veca_j}$.
    Then, as illustrated in Figure~\ref{fig:proof_fig}, we have
    \[
        \norm{\frac{\peye}{\norm{\peye}} - \frac{\pj}{\norm{\pj}}}^2 \geq \norm{\veca_i - \frac{\norm{\veca_i}}{\norm{\veca_j}} \veca_j}^2,
    \]
    and so it suffices to lower bound the right-hand side of the inequality above.
    By the triangle inequality, we have
    \begin{align*}
        \norm{\veca_i - \frac{\norm{\veca_i}}{\norm{\veca_j}} \veca_j} & \geq \norm{\veca_i - \veca_j} - \norm{\veca_j - \frac{\norm{\veca_i}}{\norm{\veca_j}} \veca_j} \\
        & = \frac{1}{1+\epsilon} \norm{\frac{\sqrt{\vol(\sets_i)} }{\norm{\barxi}}\cdot \peye - \frac{\sqrt{\vol(\sets_j)} }{\norm{\barxj}} \cdot \pj} - \left(\norm{\veca_j} - \norm{\veca_i}\right).
    \end{align*}
    Now, we have that 
    \begin{align*}
        \norm{\veca_j} - \norm{\veca_i} & = \frac{\sqrt{\vol(\sets_j)}}{(1 + \epsilon) \norm{\barxj}} \cdot \norm{\pj} - \frac{\sqrt{\vol(\sets_i)}}{(1 + \epsilon) \norm{\barxi}}\cdot   \norm{\peye} \\
        & \leq 1 - \frac{1 - \epsilon}{1 + \epsilon} \\
        & = \frac{2 \epsilon}{1 + \epsilon},
    \end{align*}
    and have by Lemma~\ref{lem:pnorm_diff} that   
    \begin{align*}
        \norm{\frac{\sqrt{\vol(\sets_i)} }{\norm{\barxi}}\cdot \peye - \frac{\sqrt{\vol(\sets_j)} }{\norm{\barxj}} \cdot \pj} & \geq \sqrt{\theta - 3 \sqrt{\Psi(\ell)}} \geq \sqrt{\theta} - \sqrt{2 \epsilon} \geq \sqrt{\theta} - 2 \epsilon.
    \end{align*}
    since $\epsilon = 4 \sqrt{\Psi} / \theta < 1$ by the assumption on $\Psi$.
    This gives us that
    \begin{align*}
        \norm{\veca_i - \frac{\norm{\veca_i}}{\norm{\veca_j}} \veca_j} & \geq \frac{\sqrt{\theta} - 2 \epsilon}{1 + \epsilon} - \frac{2\epsilon}{1 + \epsilon}  \geq \frac{1}{2} \left(\sqrt{\theta} - 4 \epsilon\right).
    \end{align*}
    Finally, we have that
    \begin{align*}
        \norm{\frac{\peye}{\norm{\peye}} - \frac{\pj}{\norm{\pj}}}^2 & \geq \frac{1}{4} \left(\sqrt{\theta} - 16 \frac{\sqrt{\Psi(\ell)}}{\theta}\right)^2 \geq \frac{\theta}{4} - 8 \sqrt{\frac{\Psi(\ell)}{\theta}},
    \end{align*}
    which completes the proof.
\end{proof}

\begin{lemma}\label{lem:pij_distance}
   It holds for different $i, j\in[k]$ that
    \[
        \norm{\peye - \pj}^2 \geq \frac{\theta^2 - 20 \sqrt{\theta\cdot \Psi(\ell)}}{16 \min\left\{\vol(\sets_i), \vol(\sets_j)\right\}}.
    \]
\end{lemma}
\begin{proof}
We assume without loss of generality that $\norm{\peye}^2 \geq \norm{\pj}^2$. Then, by Lemma~\ref{lem:p_norm} and the fact that $\|\barx^{(i)}\|^2\geq\theta$ holds for any $i\in[k]$, 
we have
    \[
        \norm{\peye}^2 \geq \left(1 - \frac{4 \sqrt{\Psi(\ell)}}{\theta}\right)\cdot \frac{\norm{\barxi}^2}{\vol(\sets_i)},
    \]
    \[
        \norm{\pj}^2 \geq \left(1 - \frac{4 \sqrt{\Psi(\ell)}}{\theta}\right)\cdot \frac{\norm{\barxj}^2}{\vol(\sets_j)},
    \]
    which implies that
    \[
        \norm{\peye}^2 \geq \frac{\theta - 4 \sqrt{\Psi(\ell)}}{\min\left\{\vol(\sets_i), \vol(\sets_j)\right\}}.
    \]
    Now, we will proceed by case distinction.
   
    Case 1: $\norm{\peye} \geq 4 \norm{\pj}$. In this case, we have
    \[
        \norm{\peye - \pj} \geq \norm{\peye} - \norm{\pj} \geq \frac{3}{4} \norm{\peye}
    \]
    and
    \begin{align*}
        \norm{\peye - \pj}^2 & \geq \frac{9}{16}\cdot \frac{\theta - 4 \sqrt{\Psi(\ell)}}{\min\left\{\vol(\sets_i), \vol(\sets_j)\right\}} \\
        & \geq \frac{\theta \left(\theta - 20 \sqrt{\Psi(\ell)/\theta}\right)}{16 \min\left\{\vol(\sets_i), \vol(\sets_j)\right\}} \\
        & = \frac{\theta^2 - 20 \sqrt{\theta\cdot  \Psi(\ell)}}{16 \min\left\{\vol(\sets_i), \vol(\sets_j)\right\}},
    \end{align*}
    since $\theta < 1$.  
   
    Case 2: $\norm{\pj} = \alpha \norm{\peye}$ for some $\alpha \in \left(\frac{1}{4}, 1\right]$.
    By Lemma~\ref{lem:norm_phat_diff}, we have
    \begin{align*}
        \left\langle \frac{\peye}{\norm{\peye}}, \frac{\pj}{\norm{\pj}} \right\rangle & \leq 1 - \frac{1}{2} \left(\frac{\theta}{4} - 8 \sqrt{\frac{\Psi}{\theta}}\right) \\
        & \leq 1 - \frac{\theta}{8} + 2 \sqrt{\frac{\Psi}{\theta}}.
    \end{align*}
    Then, it holds that 
    \begin{align*}
        \norm{\peye - \pj}^2 & = \norm{\peye}^2 + \norm{\pj}^2 - 2 \left\langle \frac{\peye}{\norm{\peye}}, \frac{\pj}{\norm{\pj}} \right\rangle \norm{\peye} \norm{\pj}   \\
        & \geq \left(1 + \alpha^2\right) \norm{\peye}^2 - 2 \left( 1 - \frac{\theta}{8} + 2 \sqrt{\frac{\Psi(\ell)}{\theta}} \right) \alpha \norm{\peye}^2 \\
        & \geq \left(1 + \alpha^2 - 2\alpha + \frac{\theta}{4} \alpha - 4 \sqrt{\frac{\Psi(\ell)}{\theta}} \alpha\right) \norm{\peye}^2 \\
        & \geq \left(\frac{\theta}{4} - 4 \sqrt{\frac{\Psi(\ell)}{\theta}}\right)\cdot \alpha\cdot  \frac{\theta - 4 \sqrt{\Psi(\ell)}}{\min\left\{\vol(\sets_i), \vol(\sets_j)\right\}} \\
        & \geq \left(\frac{\theta}{16} - \sqrt{\frac{\Psi(\ell)}{\theta}} \right) \left(\theta - 4 \sqrt{\Psi} \right)\cdot  \frac{1}{\min\left\{\vol(\sets_i), \vol(\sets_j)\right\}} \\
        & \geq \left(\frac{\theta^2}{16} -\frac{5}{4} \sqrt{\theta \Psi(\ell) } \right) \cdot \frac{1}{\min\left\{\vol(\sets_i), \vol(\sets_j)\right\}} \\
        & = \frac{\theta^2 - 20 \sqrt{\theta \Psi(\ell) }}{16 \min\left\{\vol(\sets_i), \vol(\sets_j)\right\}}
        \end{align*}
        which completes the proof.
\end{proof}

It is important to recognise that the lower bound in Lemma~\ref{lem:pij_distance} implies a condition on $\theta$ and $\Psi(\ell)$ under which $\vecp^{(i)}$ and $\vecp^{(j)}$ are well-spread.
With this, we analyse the performance of spectral clustering when fewer eigenvector are employed to construct the embedding and show that it works when the optimal clusters present a noticeable pattern.

\begin{lemma} \label{lem:cost_lower_bound_2}
Let $\{\seta_i\}_{i=1}^k$ be the output of spectral clustering with $\ell$ eigenvectors, and $\sigma$ and  $\setm_{\sigma,i}$ be defined as in \eqref{eq:defsigma} and \eqref{eq:defmset}.
If $\Psi(\ell) \leq \theta^3 / 1600$, then
\[
        \sum_{i = 1}^k \frac{\vol(\setm_{\sigma, i} \triangle \sets_i)}{\vol(\sets_i)} \leq 64 (1 + \APT) \frac{\Psi(\ell)}{\theta^2}.
    \]
\end{lemma}
\begin{proof} 
    Let us define $\setb_{ij} = \seta_i \cap \sets_j$ to be the vertices in $\seta_i$ which belong to the true cluster $\sets_j$.
    Then, we have that
    \begin{align}
        \sum_{i = 1}^k \frac{\vol(\setm_{\sigma, i} \triangle \sets_i)}{\vol(\sets_i)} & = \sum_{i = 1}^k \sum_{\substack{j = 1 \\ j \neq \sigma(i)}}^k \vol(\setb_{ij}) \left(\frac{1}{\vol(\sets_{\sigma(i)})} + \frac{1}{\vol(\sets_j)} \right) \nonumber \\
        & \leq 2 \sum_{i = 1}^k \sum_{\substack{j = 1 \\ j \neq \sigma(i)}}^k \frac{\vol(\setb_{ij})}{\min\{\vol(\sets_{\sigma(i)}), \vol(\sets_j)\}}, \label{eq:up_sym_ratio2}
    \end{align}
and that    \begin{align*}
        \lefteqn{\mathrm{COST}(\seta_1, \ldots \seta_k)}\\ & = \sum_{i = 1}^k \sum_{u \in \seta_i} \deg(u) \norm{F(u) - \vecc_i}^2 \\
        & \geq \sum_{i = 1}^k \sum_{\substack{1\leq j \leq  k \\ j \neq \sigma(i)}}   \sum_{u \in \setb_{ij}} \deg(u) \norm{F(u) - \vecc_i}^2 \\
        & \geq \sum_{i = 1}^k \sum_{\substack{1\leq j \leq k  \\ j \neq \sigma(i)}}  \sum_{u \in \setb_{ij}} \deg(u) \left(\frac{\norm{\pj - \vecc_i}^2}{2} - \norm{\pj - F(u)}^2 \right) \\
        & \geq \sum_{i = 1}^k \sum_{\substack{1\leq j \leq k \\ j \neq \sigma(i)}}  \sum_{u \in \setb_{ij}} \frac{\deg(u) \norm{\pj - \vecp^{(\sigma(i))}}^2}{8} - \sum_{i = 1}^k \sum_{\substack{1\leq j \leq k \\ j \neq i}}  \sum_{u \in \setb_{ij}} \deg(u) \norm{\pj - F(u)}^2 \\
        & \geq \sum_{i = 1}^k \sum_{\substack{1\leq j\leq k \\ j \neq \sigma(i)}}  \vol(\setb_{ij}) \frac{\norm{\pj - \vecp^{(\sigma(i))}}^2}{8} - \sum_{i = 1}^k \sum_{u \in \sets_i} \deg(u) \norm{\peye - F(u)}^2 \\
        & \geq \sum_{i = 1}^k \sum_{\substack{1\leq j\leq k \\ j \neq \sigma(i)}}   \frac{\vol(\setb_{ij})}{16\cdot  \min\{\vol(\sets_{\sigma(i)}), \vol(\sets_j)\}}\left( \theta^2 - 20\sqrt{\theta\cdot \Psi(\ell)}\right) - \sum_{i = 1}^k \sum_{u \in \sets_i} \deg(u) \norm{\peye - F(u)}^2 \\
        & \geq \frac{1}{32} \cdot\left( \sum_{i = 1}^k \frac{\vol(\setm_{\sigma, i} \triangle \sets_i)}{\vol(\sets_i)}\right) \left( \theta^2 - 20\sqrt{\theta\cdot \Psi(\ell)}\right) - \Psi(\ell),
    \end{align*}
    where the second inequality follows by the inequality of  $\norm{\veca - \vecb}^2 \geq \frac{\norm{\vecb - \vecc}^2}{2} - \norm{\veca - \vecc}^2$, the third inequality follows since $\vecc_i$ is closer to $\vecp^{(\sigma(i))}$ than $\vecp^{(j)}$, the fifth one follows from Lemma~\ref{lem:pij_distance}, and the last one follows by \eqref{eq:up_sym_ratio2}.
    
    On the other hand, since $\mathrm{COST}(\seta_1,\ldots, \seta_k) \leq \APT\cdot  \Psi(\ell)$ by Lemma~\ref{lem:cost_bound}, we have that 
    \begin{align*}
         \sum_{i = 1}^k \frac{\vol(\setm_{\sigma, i} \triangle \sets_i)}{\vol(\sets_i)} & \leq 32\cdot (1+\APT) \cdot \Psi(\ell)\cdot  \left( \theta^2 - 20\sqrt{\theta\cdot \Psi(\ell)}\right)^{-1}  \leq 64\cdot   (1+\APT) \cdot \Psi(\ell),
    \end{align*}
    where the last inequality follows by the assumption that $\Psi(\ell) \leq \theta^3 /1600$. Therefore, the statement follows.
\end{proof}

Combining this with other technical ingredients, including our developed technique for constructing the desired mapping $\sigma^*$ described in Section~\ref{sec:sc_proof_sketch}, we obtain the performance guarantee of our designed algorithm, which is summarised as follows:

\begin{theorem} \label{thm:sc_meta-graph}
Let $\graphg$ be a graph with $k$ clusters $\{\sets_i\}_{i=1}^k$ of almost balanced size,
with a $(\theta, \ell)$-distinguishable meta-graph that satisfies 
$\Psi(\ell) \leq \left(2176 (1 + \APT)\right)^{-1} \theta^3$.
Let $\{\seta_i\}_{i=1}^k$ be the output of spectral clustering with $\ell$ eigenvectors, and without loss of generality let the  optimal correspondent of $\seta_i$ be $\sets_i$. Then, it holds that
\[
        \sum_{i = 1}^k \vol\left(\seta_i \triangle \sets_{i}\right) \leq 2176 (1 + \APT) \frac{\Psi(\ell) \cdot \vol(\vertexsetg)}{k \cdot \theta^2}.
    \]
\end{theorem}
\begin{proof}
This result can be obtained by using the same technique as the one used in the proof of Theorem~\ref{thm:sc_guarantee}, but applying Lemma~\ref{lem:cost_bound} instead of Lemma~\ref{lem:total_cost} and
Lemma~\ref{lem:cost_lower_bound_2} instead of Lemma~\ref{lem:cost_lower_bound} in the analysis.
\end{proof}
 Notice that if we take $\ell = k$, then we have that $\theta = 1$ and $\Psi(\ell) \leq k / \Upsilon(k)$ which makes the guarantee in Theorem~\ref{thm:sc_meta-graph} the same as the one in Theorem~\ref{thm:sc_guarantee}.
 However, if the meta-graph corresponding to the optimal clusters is $(\theta, \ell)$-distinguishable for large $\theta$ and $\ell \ll k$, then we can have that $\Psi(\ell) \ll k / \Upsilon(k)$ and Theorem~\ref{thm:sc_meta-graph} gives a stronger guarantee than the 
 one from Theorem~\ref{thm:sc_guarantee}.

\section{Experimental Results} \label{sec:metaExperiments}
In this section we empirically evaluate the performance of spectral clustering for finding $k$ clusters while using fewer than $k$ eigenvectors.
Our results on synthetic data demonstrate that for graphs with a clear pattern of clusters,
spectral clustering with fewer than $k$ eigenvectors performs better.
This is further confirmed on real-world datasets including BSDS, MNIST, and USPS. 
The code used to produce all experimental results is available at
\begin{equation*}
\mbox{\url{https://github.com/pmacg/spectral-clustering-meta-graphs}.}
\end{equation*}
We implement the spectral clustering algorithm in Python, using the \texttt{scipy} library for computing eigenvectors, and the $k$-means algorithm from the \texttt{sklearn} library.
Our experiments on synthetic data are performed on a desktop computer with an Intel(R) Core(TM) i5-8500 CPU @ 3.00GHz processor and 16 GB RAM.
The experiments on the BSDS, MNIST, and USPS datasets are performed on a compute server with 64 AMD EPYC 7302 16-Core Processors.

\subsection{Results on Synthetic Data}
We first study the performance of spectral clustering on random graphs whose clusters exhibit a clear pattern.
Given the parameters $n \in \Z^+$, $0 \leq q \leq p \leq 1$, and some meta-graph $\graphm = (\vertexset_\graphm, \edgeset_\graphm)$ with $k$ vertices, we generate a graph with clusters $\{\sets_i\}_{i = 1}^k$, each of size $n$, as follows.
For each pair of vertices $u \in \sets_i$ and $v \in \sets_j$, we add the edge $(u, v)$ with probability $p$ if $i = j$ and with probability $q$ if $i \neq j$ 
and $(i, j) \in \edgeset_\graphm$.
The metric used for our evaluation is defined by
$
    \frac{1}{n k} \sum_{i = 1}^k \cardinality{\sets_i \cap \seta_i}$,
for the optimal matching between the output $\{\seta_i\}$ and the ground truth $\{\sets_i\}$.

In our experiments, we fix $n = 1,000$, $p = 0.01$, and consider the meta-graphs $C_{10}$ and $P_{4, 4}$, similar to those illustrated in Examples~\ref{ex:cycle} and~\ref{ex:grid}; this results in graphs with $10,000$ and $16,000$ vertices respectively.
We vary the ratio $p / q$ and the number of eigenvectors used to find the clusters.
Our experimental result, which is reported as the average score over 10 trials and shown in 
Figure~\ref{fig:sbm_results}, clearly shows that spectral clustering with fewer than $k$ eigenvectors performs better. This is particularly the case when $p$ and $q$ are close, which corresponds to the more challenging regime in the model.

\begin{figure}[ht]
\centering
\begin{subfigure}{0.35\textwidth}
    \includegraphics[width=\textwidth]{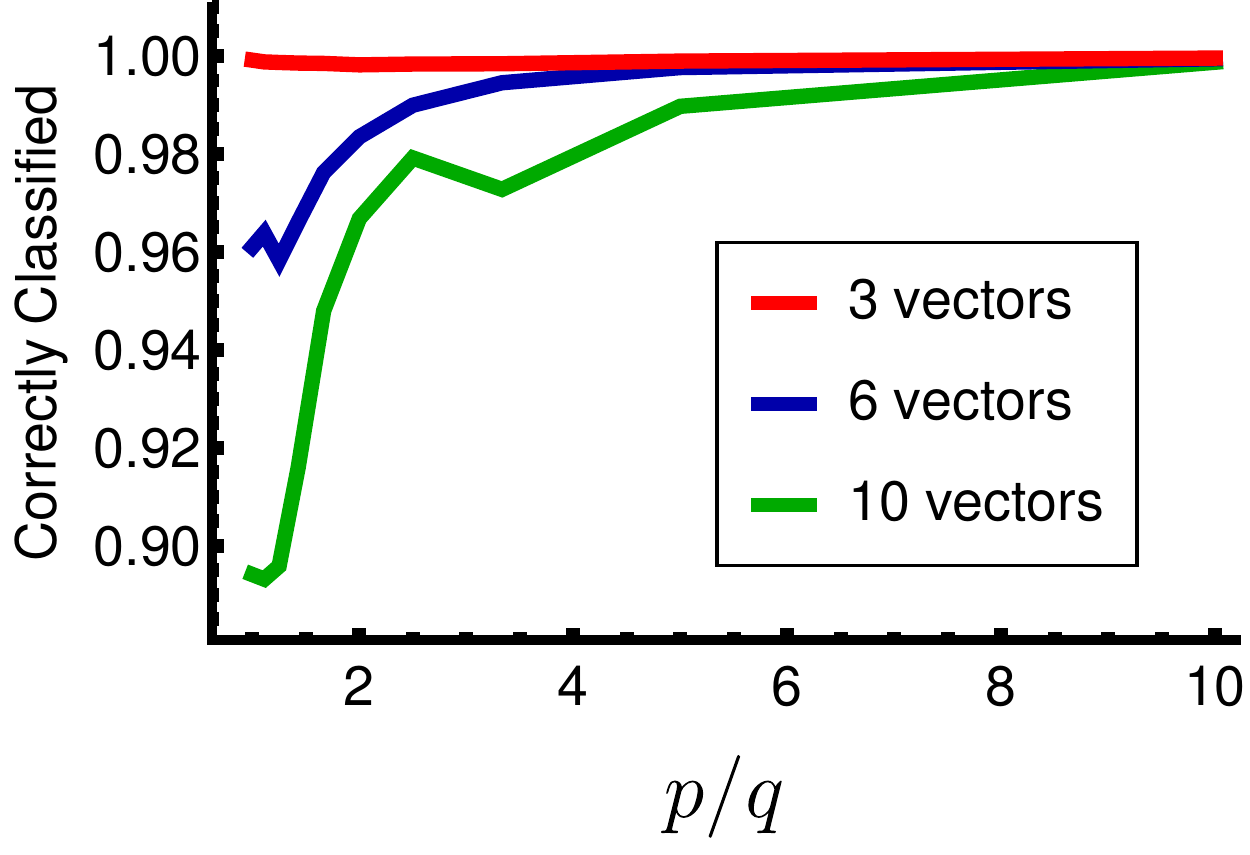}
    \caption{
    Meta-Graph $C_{10}$
    }
\end{subfigure}
\hspace{3em}
\begin{subfigure}{0.35\textwidth}
    \includegraphics[width=\textwidth]{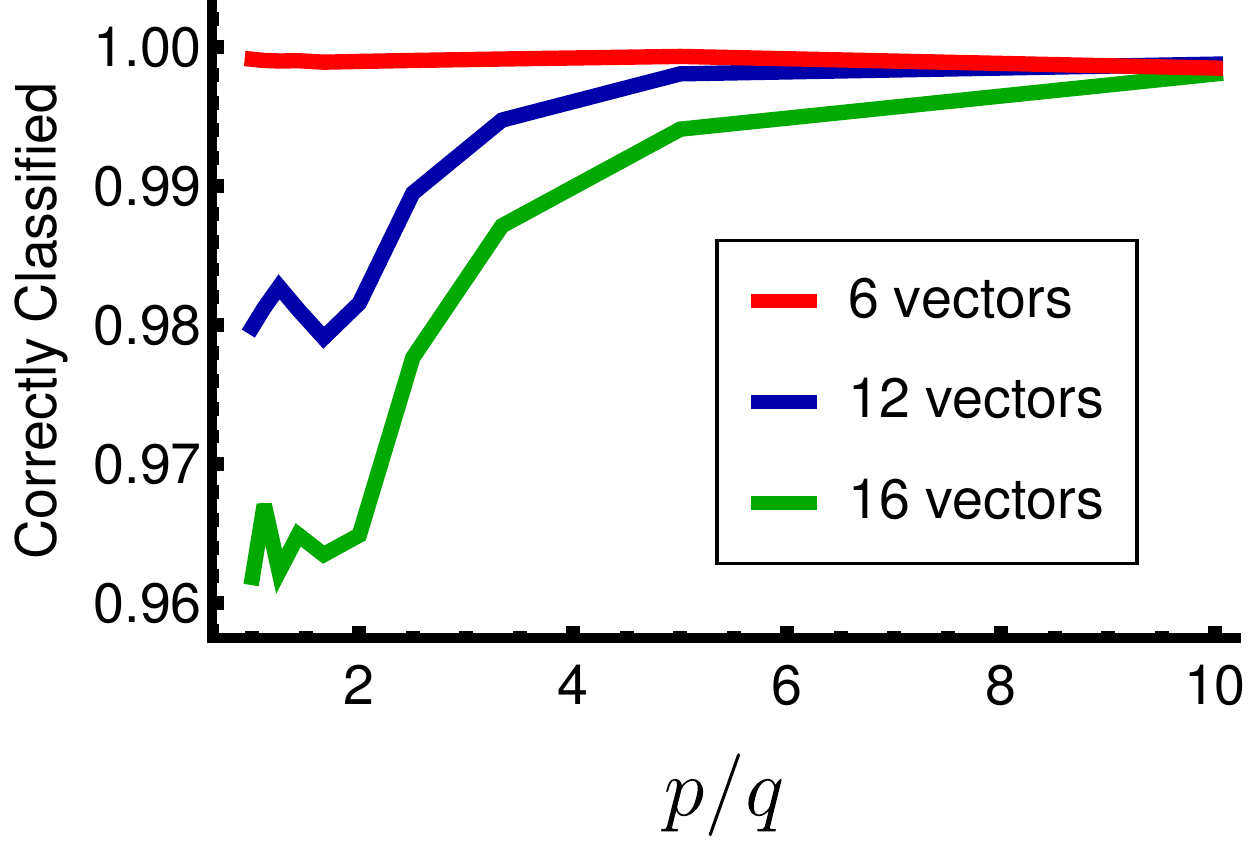}
    \caption{
        Meta-Graph $P_{4, 4}$
    }
\end{subfigure}
    \caption[The performance of spectral clustering with different numbers of eigenvectors]{
    \label{fig:sbm_results}
    A comparison of the performance of spectral clustering with different numbers of eigenvectors for clustering graphs with meta-structures $C_{10}$ and $P_{4, 4}$.
    }
\end{figure}

\subsection{Results on the BSDS Dataset}
In this experiment, we study the performance of spectral clustering for image segmentation when using different numbers of eigenvectors.
We consider the Berkeley Segmentation Data Set (BSDS)~\cite{arbelaezContourDetectionHierarchical2011}, which consists of $500$ images along with their ground-truth segmentations.
For each image, we construct a similarity graph on the pixels and take $k$ to be the number of clusters in the ground-truth segmentation\footnote{
The BSDS dataset provides several human-generated ground truth segmentations for each image.
Since there are different numbers of ground truth clusterings associated with each image, in our experiments we take the target number of clusters for a given image to be the
one closest to the median.
}.
Given a particular image in the dataset, we first downsample the image to have at most $20,000$ pixels.
Then, we represent each pixel by the point $(r, g, b, x, y)^\transpose \in \R^5$ where $r, g, b \in [1, 255]$ are the RGB values of the pixel and $x$ and $y$ are the coordinates of the pixel in the downsampled image.
We construct the similarity graph by taking each pixel to be a vertex in the graph, and for every pair of pixels $\vecu, \vecv \in \R^5$, we add an edge with weight $\exp(- \norm{\vecu - \vecv}^2 / 2 \sigma^2)$ where $\sigma = 20$.
Then we apply spectral clustering, varying the number of eigenvectors used.
We evaluate each segmentation produced with spectral clustering using the Rand Index~\cite{randObjectiveCriteriaEvaluation1971} as implemented in the benchmarking code provided along with the BSDS dataset.
For each image, this computes the average Rand Index across all of the provided ground-truth segmentations for the image.
Figure~\ref{fig:bsds_results_intro} shows two images from the dataset along with the segmentations produced with spectral clustering, and
Appendix~\ref{app:bsds} includes some additional examples.
These examples illustrate that
spectral clustering with fewer eigenvectors performs better.

We conduct the experiments on the entire BSDS dataset, and the average Rand Index of the algorithm's  output is reported in Table~\ref{tab:bsds_results}: it is clear to see that using $k/2$ eigenvectors consistently out-performs spectral clustering with $k$ eigenvectors. We further notice that, on $89\%$ of the images across the whole dataset,  using fewer than $k$ eigenvectors gives a better result than using $k$ eigenvectors.

\begin{table}[ht]
    \centering
    \begin{tabular}{cc}
    \toprule
        Number of Eigenvectors & Average Rand Index  \\
        \midrule
        $k$ & 0.71 \\
        $k / 2$ & 0.74 \\
        \textsc{Optimal} & 0.76 \\
    \bottomrule
    \end{tabular}
    \caption{The average Rand Index across the BSDS dataset for different numbers of eigenvectors.
    \textsc{Optimal}
    refers to the algorithm which runs spectral clustering with   $\ell$  eigenvectors for all possible $\ell\in[k]$ and returns the output with the highest Rand Index.}
    \label{tab:bsds_results}
\end{table}

\subsection{Results on the MNIST and USPS Datasets}
We further demonstrate the applicability of our results on the MNIST and USPS datasets~\cite{hullDatabaseHandwrittenText1994,lecunGradientbasedLearningApplied1998},   which  consist of images of hand-written digits, and the goal is to cluster the data into $10$ clusters corresponding to different digits.
In both the MNIST and USPS datasets, each image is represented as an array of grayscale pixels with values between $0$ and $255$. 
The MNIST dataset has $60,000$ images with dimensions $28 \times 28$ and the USPS dataset has $7291$ images with dimensions $16 \times 16$.
In each case, we consider each image to be a single data point in $\R^{(d^2)}$ where $d$ is the dimension of the images and construct the $k$-nearest neighbour graph for $k = 3$.
For the MNIST dataset, this gives a graph with $60,000$ vertices and $138,563$ edges and for the USPS dataset, this gives a graph with $7,291$ vertices and $16,715$ edges.
We use spectral clustering to partition the graphs into $10$ clusters.
We measure the similarity between the found clusters and the ground truth using the Adjusted Rand Index (ARI)~\cite{gatesImpactRandomModels2017},
accuracy (ACC)~\cite{randObjectiveCriteriaEvaluation1971}, and Normalised Mutual Information (NMI)~\cite{lancichinettiDetectingOverlappingHierarchical2009},
and plot the results in Figure~\ref{fig:mnist_usps}.
Our experiments show that spectral clustering with just $7$ eigenvectors gives the best performance on both datasets.

\begin{figure}[ht]
\centering
\begin{subfigure}{0.3\textwidth}
    \includegraphics[width=\textwidth]{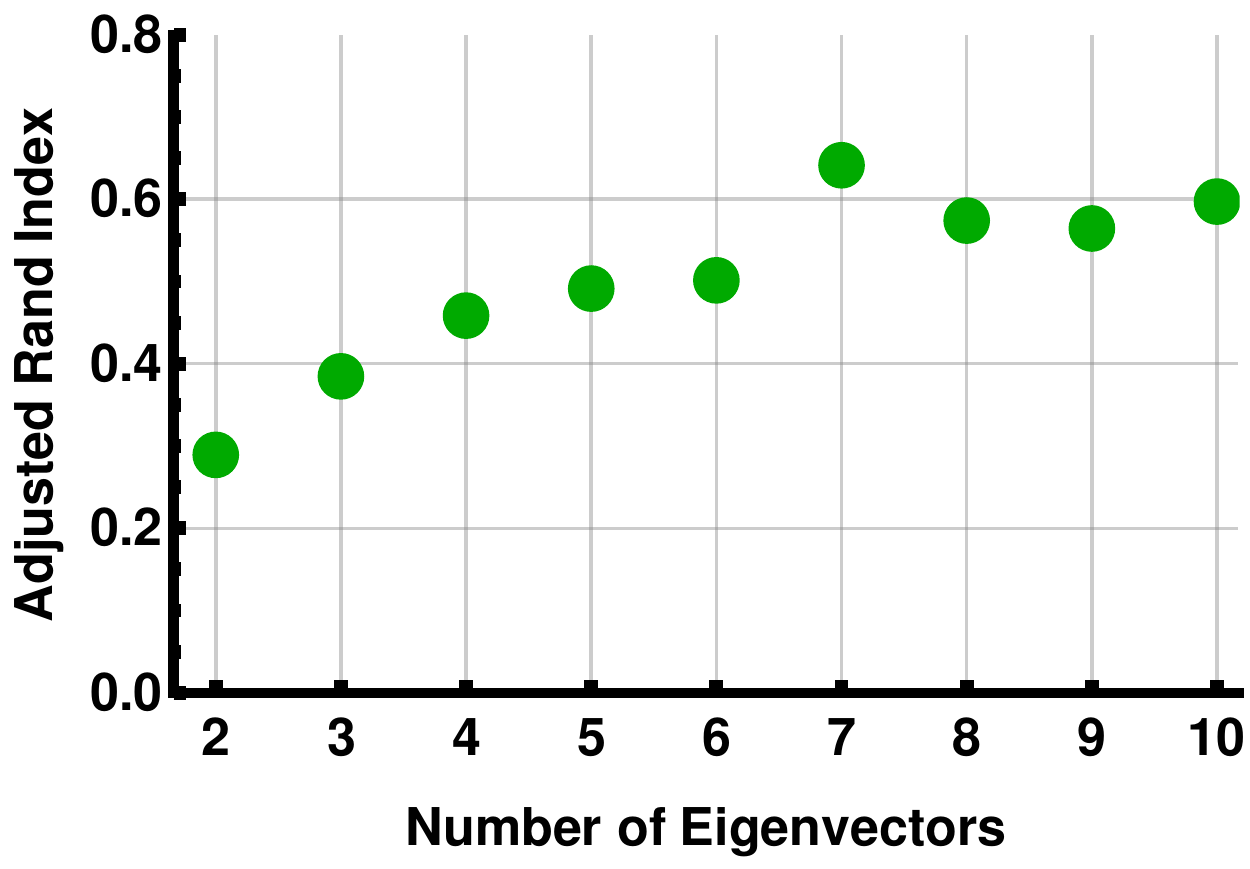}
    \caption{MNIST ARI}
\end{subfigure}
\hspace{1em}
\begin{subfigure}{0.3\textwidth}
    \includegraphics[width=\textwidth]{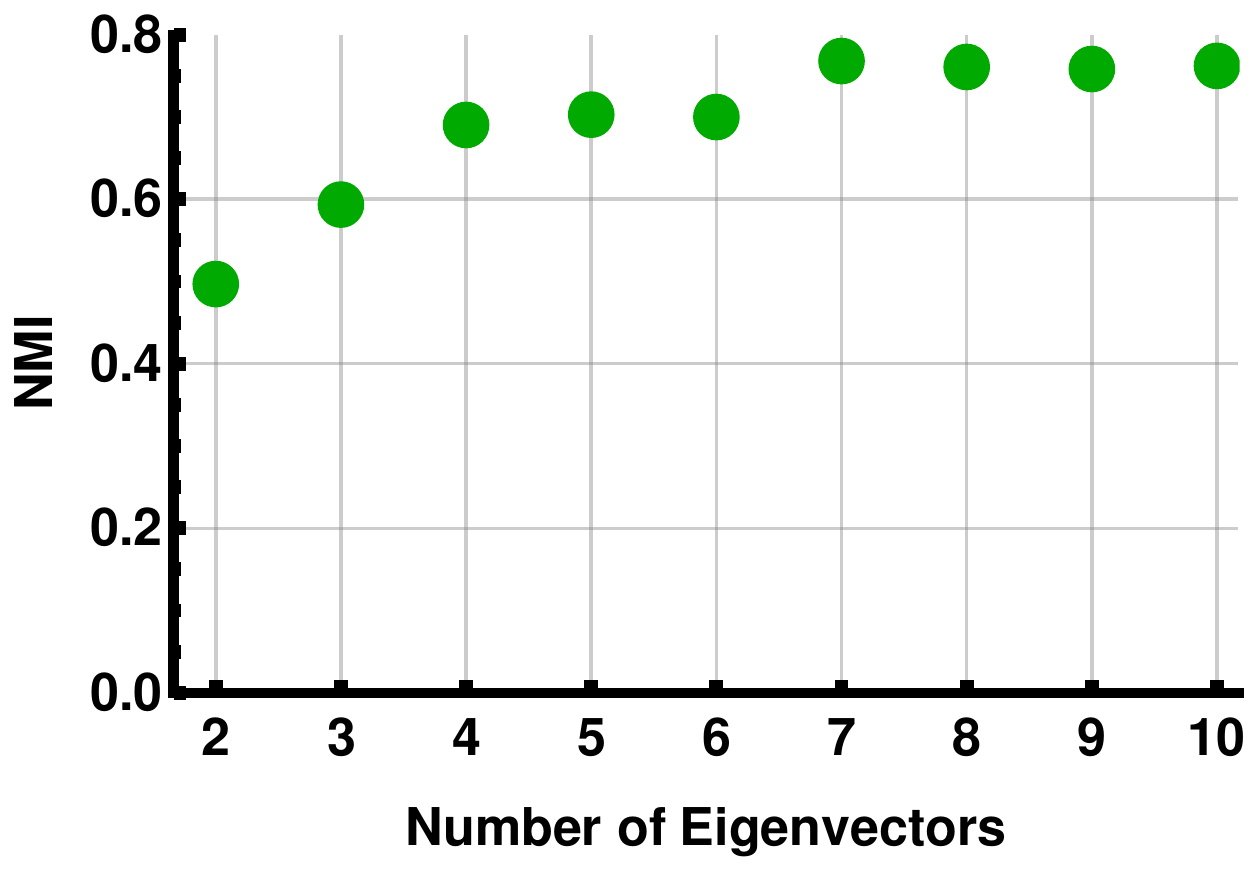}
    \caption{MNIST NMI}
\end{subfigure}
\hspace{1em}
\begin{subfigure}{0.3\textwidth}
    \includegraphics[width=\textwidth]{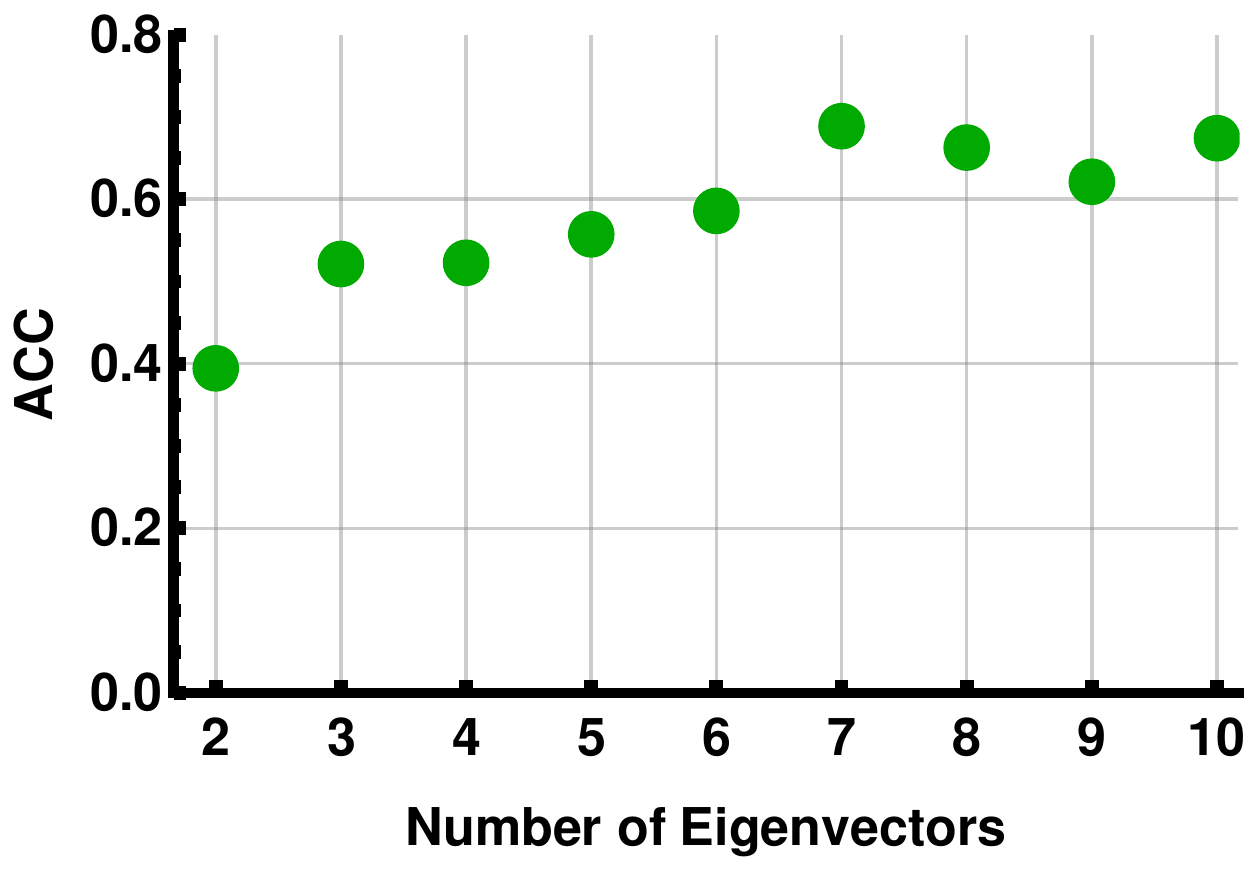}
    \caption{MNIST ACC}
\end{subfigure}
\par\bigskip
\begin{subfigure}{0.3\textwidth}
    \includegraphics[width=\textwidth]{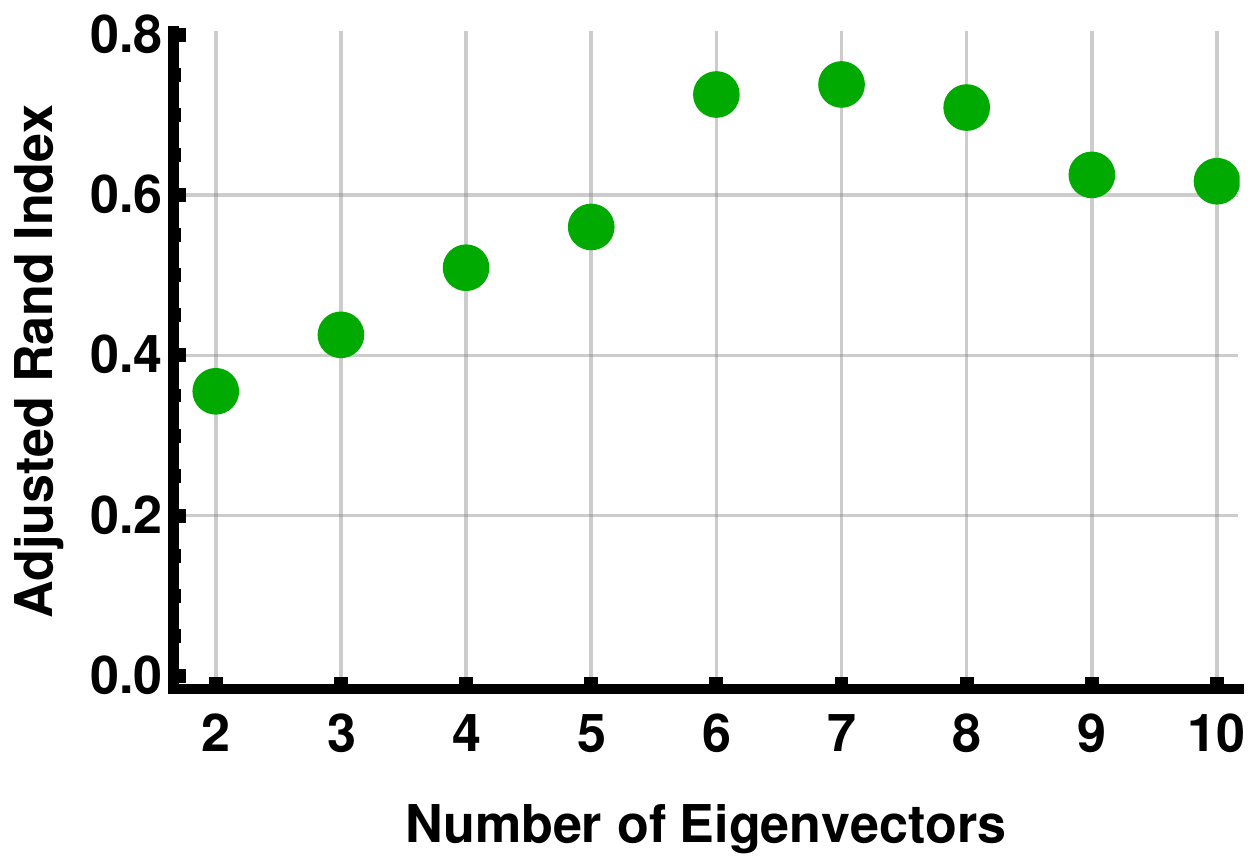}
    \caption{USPS ARI}
\end{subfigure}
\hspace{1em}
\begin{subfigure}{0.3\textwidth}
    \includegraphics[width=\textwidth]{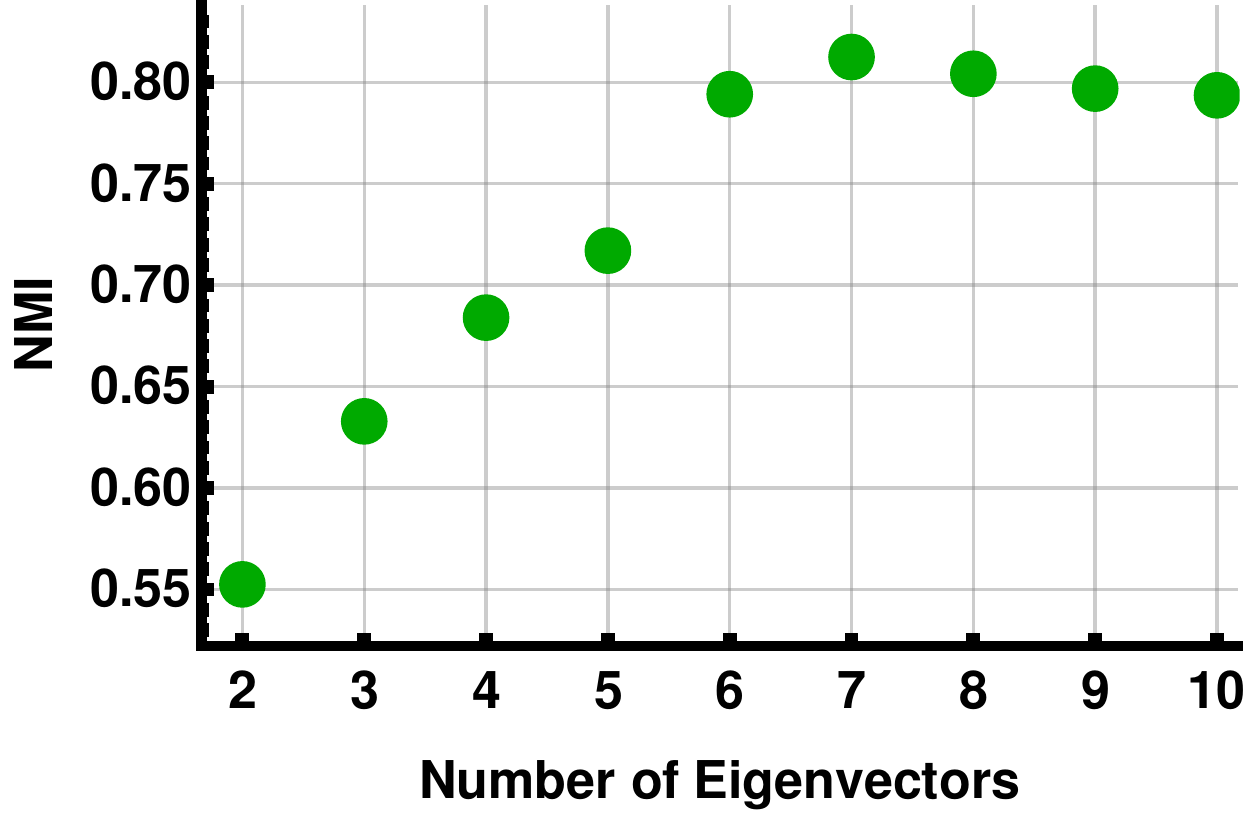}
    \caption{USPS NMI}
\end{subfigure}
\hspace{1em}
\begin{subfigure}{0.3\textwidth}
    \includegraphics[width=\textwidth]{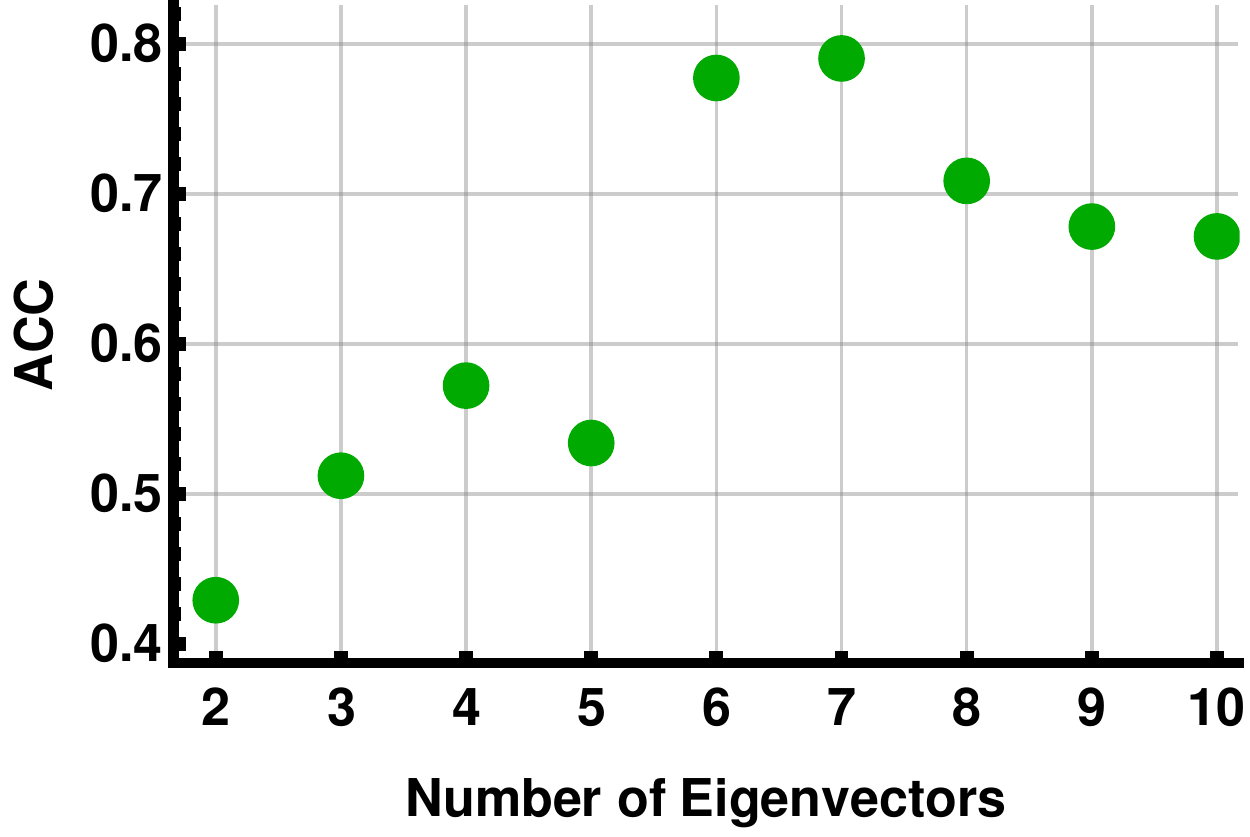}
    \caption{USPS ACC}
\end{subfigure}
    \caption[Experimental results on the MNIST and USPS datasets]{
    \label{fig:mnist_usps}
    Experimental results on the MNIST and USPS datasets. These experiments show that spectral clustering  with $7$ eigenvectors gives the best partition of the input into $10$ clusters.}
\end{figure}

\section{Future Work}
Our work leaves a number of interesting questions for future research.
For spectral clustering, the only non-trivial assumption remaining in our analysis is that the optimal clusters have almost balanced size.
It is unclear whether, under the regime of $\Upsilon(k) = \Omega(1)$, this condition could be eventually removed, or if there's some hard instance showing that our analysis is tight.
For spectral clustering with fewer eigenvectors, our presented work is merely the starting point, and leaves many open questions.
For example, although one can enumerate the number of used eigenvectors from $1$ to $k$ and take the clustering with the minimum $k$-way expansion,
we are interested to know whether the optimal number of eigenvectors
can be computed directly, and rigorously analysed for different graph instances.
We believe that the answers to these questions would not only significantly advance our understanding of spectral clustering, but also, as suggested in our experimental studies, have widespread applications in analysing real-world datasets.

\bibliographystyle{alpha}
\bibliography{arxiv_references}

\appendix
\newpage
\section{Examples from the BSDS Dataset} \label{app:bsds}
Figures~\ref{fig:bsds_app_0}~and~\ref{fig:bsds_results_appendix} give some additional examples of our results from the BSDS dataset.
These examples further illustrate that
spectral clustering with fewer than $k$ eigenvectors performs better than spectral clustering with $k$ eigenvectors on the BSDS image segmentation dataset.

\begin{figure*}[ht]
    \centering
    \begin{subfigure}[t]{0.25\textwidth}
        \includegraphics[width=\textwidth]{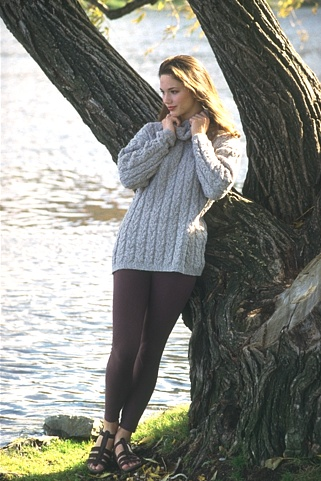}
        \caption{Original Image}
    \end{subfigure}
    \hspace{1.5em}
    \begin{subfigure}[t]{0.25\textwidth}
        \includegraphics[width=\textwidth]{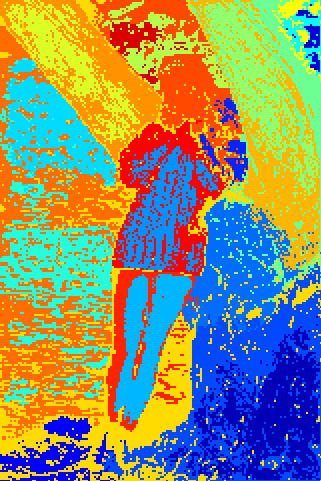}
        \caption{Segmentation into $24$ clusters with $8$ eigenvectors; Rand Index $0.82$.}
    \end{subfigure}
    \hspace{1.5em}
    \begin{subfigure}[t]{0.25\textwidth}
        \includegraphics[width=\textwidth]{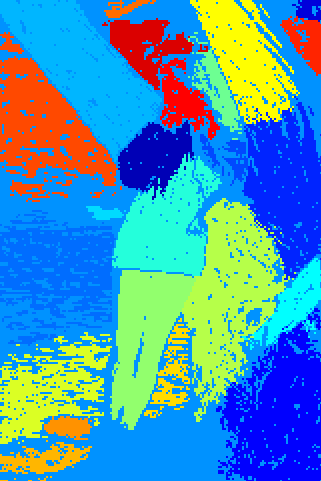}
        \caption{Segmentation into $24$ clusters with $24$ eigenvectors; Rand Index $0.77$.}
    \end{subfigure}
    \par\bigskip
    \begin{subfigure}[t]{0.25\textwidth}
        \includegraphics[width=\textwidth]{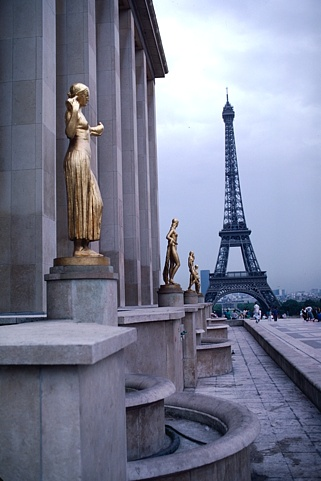}
        \caption{Original Image}
    \end{subfigure}
    \hspace{1.5em}
    \begin{subfigure}[t]{0.25\textwidth}
        \includegraphics[width=\textwidth]{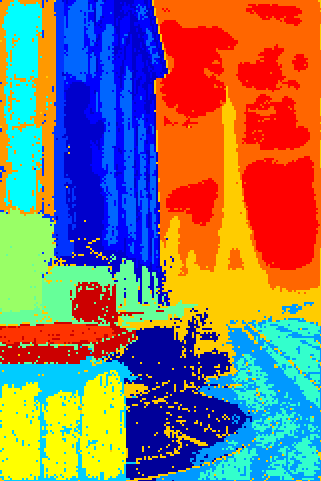}
        \caption{Segmentation into $18$ clusters with $6$ eigenvectors; Rand Index $0.77$.}
    \end{subfigure}
    \hspace{1.5em}
    \begin{subfigure}[t]{0.25\textwidth}
        \includegraphics[width=\textwidth]{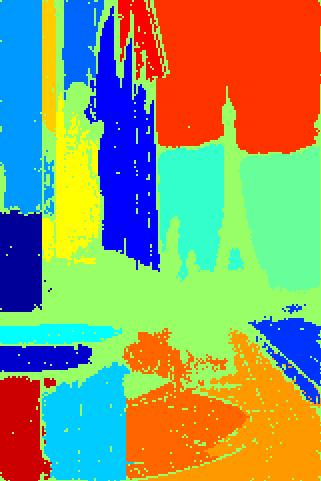}
        \caption{Segmentation into $18$ clusters with $18$ eigenvectors; Rand Index $0.74$.}
    \end{subfigure}
    \caption[Additonal examples from the BSDS dataset]{Examples of the segmentations produced with spectral clustering on the BSDS dataset.} 
    \label{fig:bsds_app_0}
\end{figure*}

\begin{figure*}[t]
    \centering
    \begin{subfigure}[t]{0.3\textwidth}
       \includegraphics[width=\textwidth]{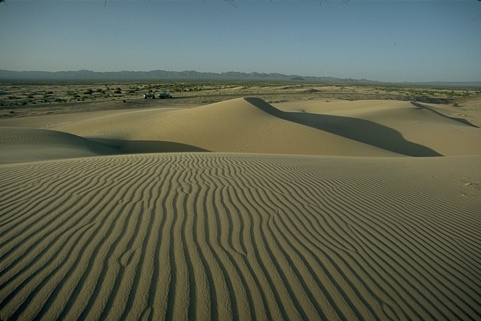}
       \caption{Original Image}
    \end{subfigure}
    \hspace{1em}
    \begin{subfigure}[t]{0.3\textwidth}
       \includegraphics[width=\textwidth]{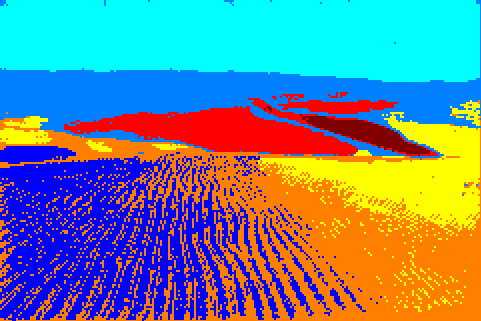}
       \caption{Segmentation into $7$ clusters with $3$ eigenvectors; Rand Index $0.76$.}
    \end{subfigure}
    \hspace{1em}
    \begin{subfigure}[t]{0.3\textwidth}
       \includegraphics[width=\textwidth]{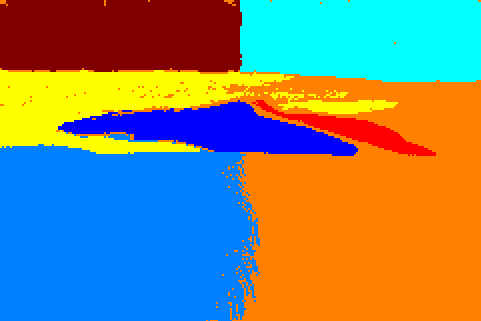}
       \caption{Segmentation into $7$ clusters with $7$ eigenvectors; Rand Index $0.74$.}
    \end{subfigure}
    \par\bigskip
    \begin{subfigure}[t]{0.3\textwidth}
       \includegraphics[width=\textwidth]{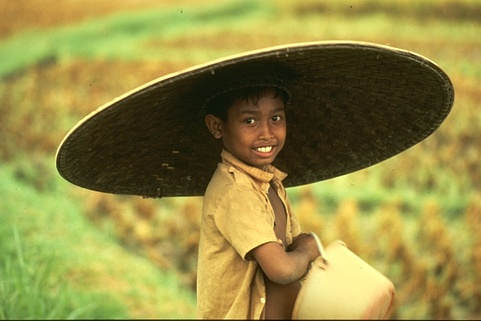}
       \caption{Original Image}
    \end{subfigure}
    \hspace{1em}
    \begin{subfigure}[t]{0.3\textwidth}
       \includegraphics[width=\textwidth]{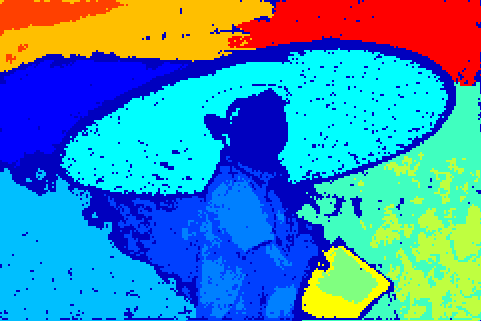}
       \caption{Segmentation into $13$ clusters with $8$ eigenvectors; Rand Index $0.80$.}
    \end{subfigure}
    \hspace{1em}
    \begin{subfigure}[t]{0.3\textwidth}
       \includegraphics[width=\textwidth]{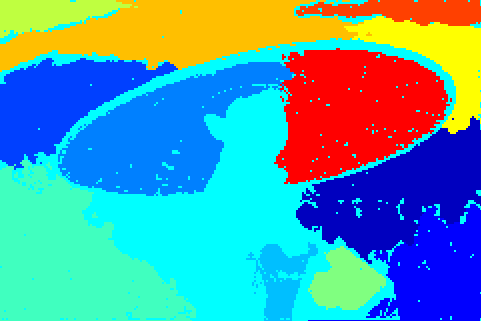}
       \caption{Segmentation into $13$ clusters with $13$ eigenvectors; Rand Index $0.77$.}
    \end{subfigure}
    \par\bigskip
    \begin{subfigure}[t]{0.3\textwidth}
       \includegraphics[width=\textwidth]{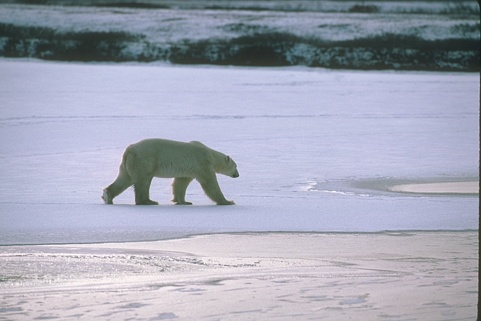}
       \caption{Original Image}
    \end{subfigure}
    \hspace{1em}
    \begin{subfigure}[t]{0.3\textwidth}
       \includegraphics[width=\textwidth]{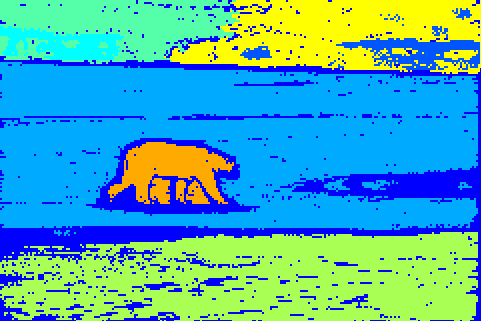}
       \caption{Segmentation into $8$ clusters with $5$ eigenvectors; Rand Index $0.86$.}
    \end{subfigure}
    \hspace{1em}
    \begin{subfigure}[t]{0.3\textwidth}
       \includegraphics[width=\textwidth]{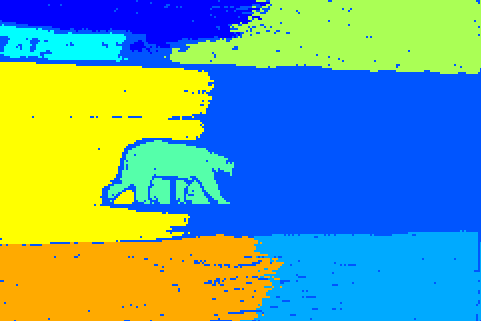}
       \caption{Segmentation into $8$ clusters with $8$ eigenvectors; Rand Index $0.79$.}
    \end{subfigure}
    \par\bigskip
    \begin{subfigure}[t]{0.3\textwidth}
       \includegraphics[width=\textwidth]{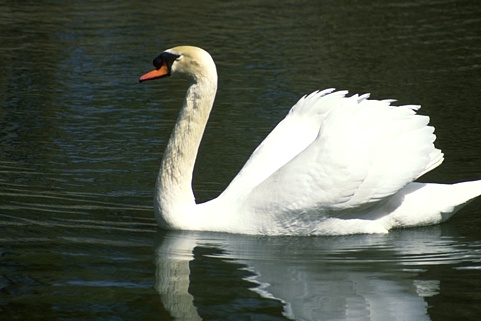}
       \caption{Original Image}
    \end{subfigure}
    \hspace{1em}
    \begin{subfigure}[t]{0.3\textwidth}
       \includegraphics[width=\textwidth]{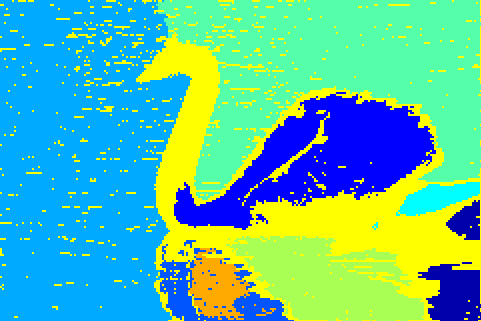}
       \caption{Segmentation into $9$ clusters with $7$ eigenvectors; Rand Index $0.69$.}
    \end{subfigure}
    \hspace{1em}
    \begin{subfigure}[t]{0.3\textwidth}
       \includegraphics[width=\textwidth]{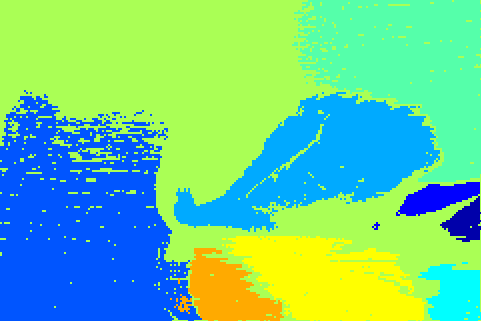}
       \caption{Segmentation into $9$ clusters with $9$ eigenvectors; Rand Index $0.61$.}
    \end{subfigure}
    \caption[Additional examples from the BSDS dataset]{Examples of the segmentations produced with spectral clustering on the BSDS dataset.} 
    \label{fig:bsds_results_appendix}
\end{figure*}

\end{document}